\documentclass[aps,superscriptaddress,twocolumn,twoside,floatfix,pra,nofootinbib,a4paper]{revtex4-2}

\usepackage{times}
\usepackage{dsfont}
\usepackage{amsfonts}
\usepackage{amsmath}
\usepackage{amssymb}
\usepackage{amsthm}
\usepackage{multirow}
\usepackage[normalem]{ulem}
\usepackage[T1]{fontenc}
\usepackage{float}
\usepackage{mathtools}
\usepackage{bm}
\usepackage[UKenglish,cleanlook]{isodate}   
\usepackage{mathrsfs}
\usepackage{amsthm}
\usepackage{comment}

\usepackage[ruled,vlined]{algorithm2e}

\newtheorem{theorem}{Theorem}
\theoremstyle{plain}

\providecommand{\theoremname}{Theorem}
\usepackage{physics}

\usepackage{xcolor}
\definecolor{carmine}{RGB}{150,0,24}

\newcommand{\headeritem}[1]{\item \textbf{#1}\quad}
\usepackage{natbib}
\usepackage[colorlinks=true,linkcolor=blue,citecolor=magenta,urlcolor=blue]{hyperref}

\usepackage{caption}
\usepackage{graphicx}
\usepackage{tikz-cd}

\begin{document}
	

\title{Quantum coherence as randomness under classical control}

\author{Gabriele Cobucci}\thanks{These authors contributed equally}
\affiliation{Physics Department and NanoLund, Lund University, Box 118, 22100 Lund, Sweden.}

\author{Nicola D'Alessandro}\thanks{These authors contributed equally}
\affiliation{Physics Department and NanoLund, Lund University, Box 118, 22100 Lund, Sweden.}

\author{Raphael Brinster}		
\affiliation{Institut f\"ur Theoretische Physik III, Heinrich-Heine-Universit\"at D\"usseldorf,
	Universit\"atsstra{\ss}e 1, 40225 D\"usseldorf, Germany}
	
\author{Alexander Bernal}
\affiliation{Physics Department and NanoLund, Lund University, Box 118, 22100 Lund, Sweden.}	
	
\author{Nikolai Wyderka}		
\affiliation{Institut f\"ur Theoretische Physik III, Heinrich-Heine-Universit\"at D\"usseldorf,
	Universit\"atsstra{\ss}e 1, 40225 D\"usseldorf, Germany}

\author{Armin Tavakoli}		
\affiliation{Physics Department and NanoLund, Lund University, Box 118, 22100 Lund, Sweden.}

\begin{abstract}
Quantum states that do not commute exhibit coherence, but only when the device preparing them is assumed to be unaffected by classical parameters inaccessible to the experimenter. Accounting for such hidden classical control leads to a stronger notion of coherence. We establish its significance through the  task of randomness certification. It is shown that a preparation device is coherent if and only if it can be used to certify randomness in a prepare-and-measure experiment. This operational interpretation motivates the problem of certifying coherence in the presence of hidden classical control.  We  obtain a complete characterisation of coherence through a hierarchy of semidefinite programs. Then, we introduce a practical semidefinite programming approach that achieves useful accuracy while remaining computationally efficient even for devices preparing  many, potentially high-dimensional, quantum states. For the important  case of qubits, we further exploit  connections with the theory of joint measurability to obtain highly accurate coherence characterisation that scales to more than one thousand qubit states. Our results give  operational meaning to quantum coherence under hidden classical control and provide a versatile toolbox for its certification.
\end{abstract}

\date{\today}

\maketitle

\section{Introduction}

Superposition lies at the heart of quantum theory and constitutes an essential resource for quantum science and technology. However, certifying superposition crucially depends on the assumptions made about the system. In its simplest formulation, a quantum state that is not diagonal in a chosen  reference basis --- whose elements define the incoherent states --- is said to be in superposition and is therefore called coherent. Quantifying, transforming, and interpreting such coherence has received much attention \cite{Baumgratz2014, Streltsov2017}. This perspective, however, comes with the conceptual limitation that coherence is basis dependent. It raises the question of whether coherence can be addressed without specifying a privileged basis. 

Since any quantum state is diagonal in its own eigenbasis, basis-independent coherence cannot be a property of a single state but  arises from relations between multiple states. A set of states admits a common diagonalising basis if and only if the states commute pairwise. Consequently, non-commuting states exhibit coherence, which has been quantified using tools from resource theory \cite{Designolle2021}. However, this conclusion presumes complete knowledge of the preparation process. If part of an incoherent preparation mechanism is inaccessible to the user, the states available to them may nevertheless be non-commuting, leading to a false certification of coherence.

This limitation motivates a more general description of the preparation process. In quantum foundations and quantum information, it is common to allow quantum devices to be influenced by classical variables that are inaccessible to the user, often referred to as hidden variables. Such variables play a central role in foundational scenarios such as Bell tests \cite{Brunner2014}, quantum steering \cite{Uola2020}, and prepare-and-measure experiments \cite{PMreview}. They also arise naturally in quantum information protocols, where they can represent remote classical influences or side information affecting the operation of a quantum device, for instance in entanglement-assisted classical communication \cite{Buhrmann2010} and quantum random number generation \cite{Pironio20102, Shalm2021}. Since classical control is comparatively inexpensive and can substantially enlarge the range of behaviours compatible with a classical description (see, e.g., \cite{Bowles2014,Vicente2017}), it is natural to ask how the notion of coherence changes when such hidden classical information is taken into account. 

Hidden classical control of a preparation device means that a parameter, $\lambda$, inaccessible to the user, can stochastically modify the device's programming by following a predetermined rule. Ref~\cite{Cobucci2026} recently applied this perspective to the notion of coherence.  Suppose that $\lambda$ determines the basis in which the preparation device generates a set of commuting states. By marginalising over $\lambda$, this incoherent process can appear coherent to a user unaware of the hidden classical control. As an example, consider noisy preparations of the positive eigenstates of the Pauli $X$ and $Z$ observables, $\rho_1=v \ketbra{0}+\frac{1-v}{2}\mathds{1}$ and $\rho_2=v \ketbra{+}+\frac{1-v}{2}\mathds{1}$, where $\ket{+}=\frac{1}{\sqrt{2}}\left(\ket{0}+\ket{1}\right)$ and $x\in\{1,2\}$ labels the desired preparation $\rho_x$. These states commute only when the visibility parameter $v\in[0,1]$ is zero. Nevertheless, a simple incoherent process can simulate them for a much wider range of visibilities. Let $\lambda\in\{1,2\}$ be uniformly distributed. When $\lambda=1$, the device generates commuting states whose Bloch vectors are $\pm\frac{1}{\sqrt{2}}(1,0,-1)$ for $x=1,2$, respectively. When $\lambda=2$, the device always produces the state with Bloch vector $\frac{1}{\sqrt{2}}(1,0,1)$. After marginalising $\lambda$, the states observed by the user coincide with $\{\rho_1,\rho_2\}$ when $v=\frac{1}{\sqrt{2}}$ (see Fig. \ref{Fig:simulation}). Thus, for $0 < v\leq \frac{1}{\sqrt{2}}$, the observed states do not commute yet they admit a description which, conditioned on $\lambda$, is based entirely on incoherent preparations.

\begin{figure}[t!]
	\centering
	\includegraphics[width=1\linewidth]{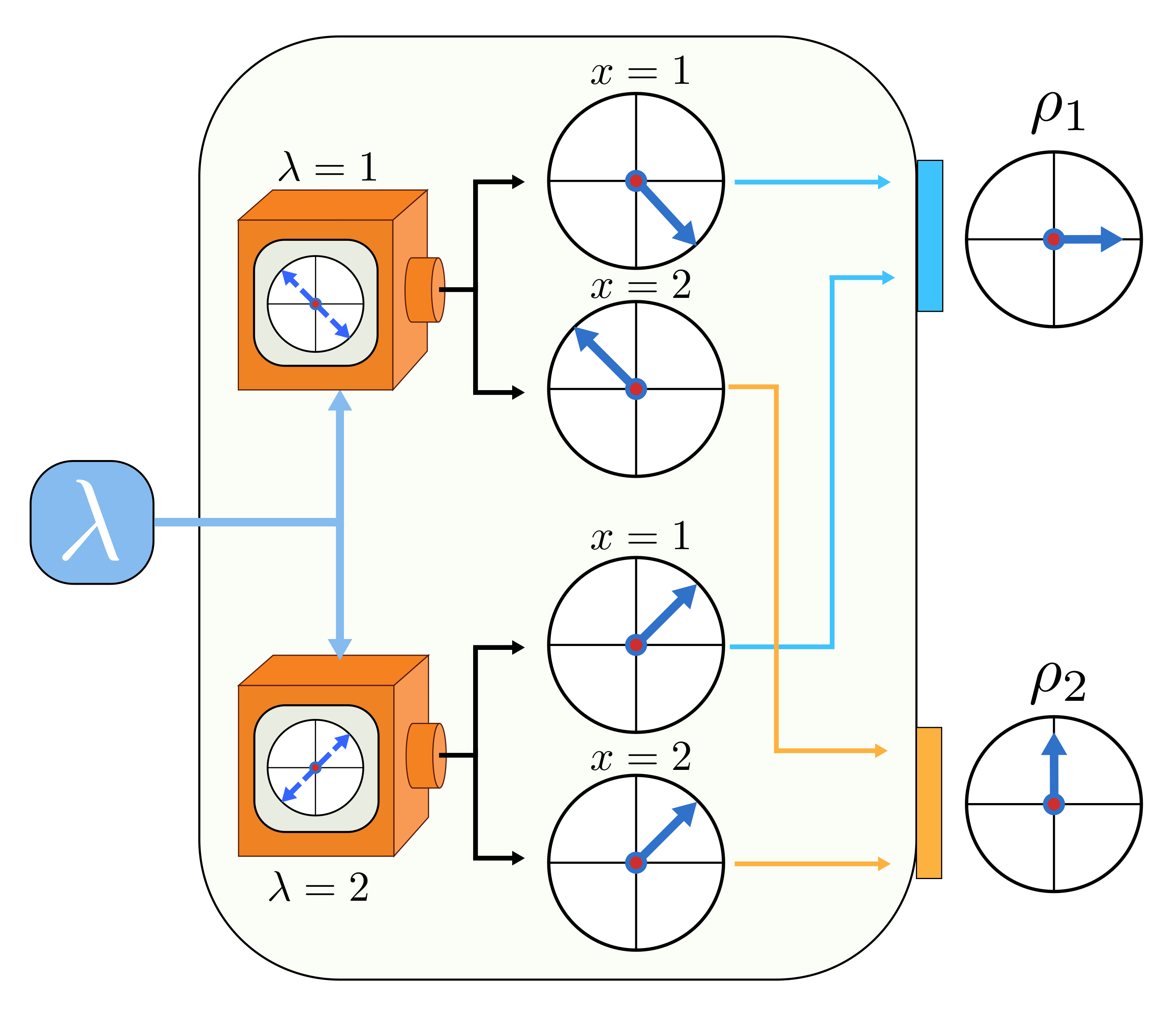}
	\caption{Incoherent model simulating the non-commuting states $\rho_1=v \ketbra{0}+\frac{1-v}{2}\mathds{1}$ and $\rho_2=v \ketbra{+}+\frac{1-v}{2}\mathds{1}$ for $v=\frac{1}{\sqrt{2}}$, where $\ket{+}=\frac{1}{\sqrt{2}}\left(\ket{0}+\ket{1}\right)$. A  parameter $\lambda=\lbrace 1,2 \rbrace$ uniformly chooses a preparation device which generates commuting states whose Bloch vectors are represented in the figure. The diagonal in each box denotes the basis in which each device operates.}
	\label{Fig:simulation}
\end{figure}

A central question is whether coherence under classical control admits an operational interpretation relevant to quantum information science. We answer this in the affirmative by considering the task of quantum random number generation.
Previous works have identified quantum coherence relative to a fixed, trusted measurement basis as the resource underlying intrinsic randomness~\cite{Yuan2015IntrinsicRandomness, HayashiZhu2018SecureRandomness, Ma2019CoherenceQRNG}. Here, we remove this trust assumption. The relevant basis is unspecified and may depend on hidden classical information available to the adversary. Specifically, we consider the standard prepare-and-measure scenario, in which the preparation and measurement devices may share hidden classical information that governs their operation. We prove that a preparation device is coherent if and only if it enables the certification of randomness against an adversary with access to this hidden information. Thus, coherence under classical control is the decisive resource that enables device-independent randomness certification: it characterises those preparation devices capable of producing outcomes that cannot be predicted from the classical information governing the devices.

Having established the operational significance of coherence, we turn to its certification. We develop general and versatile methods  using relaxation techniques based on semidefinite programming \cite{Tavakoli2024}. We  prove that coherence admits a complete characterisation through a hierarchy of semidefinite programs (SDPs). Although this necessary and sufficient condition is conceptually appealing, it is less suitable for certifying real-life quantum preparation devices. To address the practical setting, we therefore introduce an alternative SDP condition that is sufficient for coherence and can be applied to arbitrary sets of states. This method achieves useful accuracy while remaining computationally efficient even for large numbers of preparations and for systems with Hilbert space dimensions well into the double-digit regime. For the important special case of qubit states, we further identify connections between coherence certification and results from the literature on joint measurability \cite{Guhne2023}. Exploiting these connections allows us to analyse coherence using tools originally developed for measurement compatibility. This leads to highly accurate coherence characterisation for sets of more than one thousand qubit preparations, while also enabling certification of the absence of coherence. Finally, we extend our methods from preparation devices to quantum channels and show how to certify whether a channel breaks coherence or is able to preserve it. This provides a practical characterisation of what we define as coherence-breaking channels.

\section{Preliminaries \& basic results}
	Consider a preparation device that emits a set of $N$ quantum states, $\mathcal{E}=\{\rho_x\}_{x=1}^N$, where each state has dimension $d$. For simplicity, we have chosen the set to be of finite size but it can equally well be uncountably infinite. We are interested in whether the states $\mathcal{E}$ are coherent when we allow the preparation device to have hidden classical parameters, i.e.~the device can be classically pre-programmed with a hidden variable, $\lambda$. Such models were formalised in Ref~\cite{Cobucci2026}: they can sample $\lambda$ from a probability density $q(\lambda)$ and associate to it a preparation device that takes $x$ as input and produces a set of states $\{\tau_{x,\lambda}\}_x$ which is incoherent in the standard sense that it is jointly diagonalisable. Thus, it holds that  
	\begin{equation}\label{commutation}
		\left[\tau_{x,\lambda},\tau_{x',\lambda}\right]=0, \quad \forall x, x'.
	\end{equation}
	Since $\lambda$ is not accessible to the user, the state seen is not the commuting collection $\{\tau_{x,\lambda}\}_{x,\lambda}$ but instead the average state 
	\begin{equation}\label{classicalmodel}
		\rho_x=\int d\lambda\ q(\lambda) \tau_{x,\lambda},
	\end{equation}
	which may be non-commuting. We let  $\mathcal{C}$ label the set of all $\mathcal{E}$ that have a decomposition \eqref{classicalmodel}  satisfying Eq~\eqref{commutation}. If $\mathcal{E}\in \mathcal{C}$ we say that the $\mathcal{E}$ is incoherent, whereas otherwise it is coherent. 
	
	Let us highlight some  basic properties of the set $\mathcal{C}$. Firstly, it is convex and closed: given two incoherent sets $\mathcal{E} = \lbrace \rho_x \rbrace_{x}$ and $\mathcal{E}' = \lbrace \sigma_y \rbrace_{y}$, then also $\lbrace p \,\rho_x + (1-p) \, \sigma_y\rbrace_{x,y} \in \mathcal{C}$  $\forall p \in [0,1]$. Secondly, since coherence is a relational property it follows that for any single state (i.e.~when $|\mathcal{E}|=1$) it holds that $\mathcal{E}\in\mathcal{C}$. Thirdly, if all states in $\mathcal{E}$ are pure, the incoherent model reduces to standard commutation, namely $[\rho_x,\rho_{x'}]=0$ $\forall x,x'$. Finally, consider that we have a set of states $\{\rho_x\}_x$ and that we obtain another set of states $\{\rho'_{z}\}_z$ by classically wiring the formers, i.e.~$\rho_z'=\sum_x p(x\lvert z)\rho_x$. Then, it follows that  $\{\rho'_{z}\}_z\in\mathcal{C}$ if $\{\rho_{x}\}_x\in\mathcal{C}$. Indeed, decomposing $\rho_x$ as in Eq~\eqref{classicalmodel}, we can define the states $\varphi_{z,\lambda} \equiv \sum_x p(x|z) \tau_{x,\lambda}$ and notice that they commute, i.e. $[\varphi_{z,\lambda},\varphi_{z',\lambda}] = 0$ $\forall z,z'$. Therefore, since each state $\rho'_{z}$ admits a decomposition of the form in \eqref{classicalmodel} using the states $\lbrace \varphi_{z,\lambda}\rbrace_{z}$, it follows that $\lbrace \rho'_{z} \rbrace_{z} \in \mathcal{C}$.

Next, we discuss the quantification of the coherence of $\mathcal{E}$. One possible approach that is common in resource theories is a generalised robustness measure, namely $v_\text{worst}(\mathcal{E}) = \max\left\{0\leq v\leq 1 \quad \text{s.t.}\quad \left\{v\rho_x+(1-v)\sigma_x\right\}_x \in\mathcal{C}\right\}$,
where $\{\sigma_x\}_x$ is an arbitrary set of states. It represents the largest possible visibility that $\mathcal{E}$ can have  in order to eliminate all coherence. While this is mathematically convenient, it is physically less natural since it identifies only the most destructive noise possible, regardless of its physical justification. In view of that, it is often more motivated to  consider an unbiased (isotropic) noise model, which corresponds to selecting $\sigma_x=\frac{1}{d}\openone$. This leads to the coherence quantifier 
\begin{equation}\label{robustness_iso}
	v^*(\mathcal{E}) = \max \left\{0\leq v\leq 1 \quad \text{s.t.}\quad \big\{v\rho_x+\frac{1-v}{d}\openone \big\}_x \in\mathcal{C}\right\}.
\end{equation}
This quantifier is faithful by construction: one has $v^*=1$ if and only if $\mathcal{E}\in\mathcal{C}$. Note that a smaller $v^*$ means a larger amount of coherence.

\section{Coherence as a resource for certified randomness}
We now establish the operational meaning of coherence under classical control by considering its role for the task of  quantum random number generation. Quantum random number generation is a well-established quantum technology \cite{Herrero2017, Ma2016} with a broad range of applications \cite{Amer2025}. A common experimental setting is the prepare-and-measure scenario, in which a user selects an input to a preparation device and sends its output to a measurement device that produces an outcome. In the semi-device-independent paradigm, which is one of the main avenues for quantum random number generation \cite{PMreview}, the internal workings of the preparation and measurement devices are left uncharacterised and may depend on a hidden classical variable $\lambda$, as illustrated in Fig.~\ref{fig_QRNG}.

Most works consider classical side information, although security against quantum side information has also been studied. The classical model is well motivated in many practical settings. In particular, it better captures scenarios in which the devices are not malicious but are imperfectly controlled \cite{Pironio2013, Knill2020}. Moreover, practical implementations can be compact and self-contained, making it feasible to isolate their quantum components and prevent quantum leakage \cite{Wang2023, Tebyanian2024}. Alternatively, limitations on quantum-memory times can, under suitable protocol assumptions, effectively reduce quantum side information to classical information \cite{Pironio2023b}. We will prove that a preparation device that provides the user with the states $\mathcal{E}$ is capable of certifying randomness in these experiments if and only if $\mathcal{E}$ is coherent.

\begin{figure}[t!]
	\centering
	\includegraphics[width=0.9\linewidth]{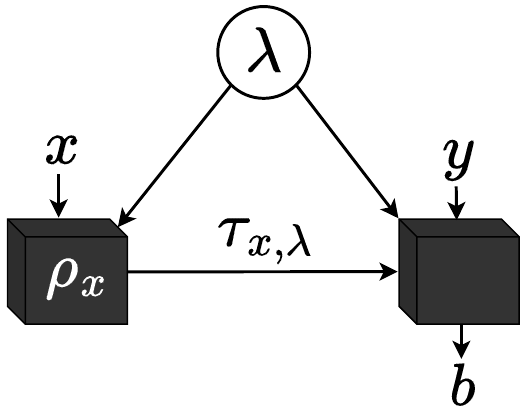}
	\caption{Quantum random number generation in the prepare and measure scenario. For input $x$, the hidden variable $\lambda$ determines the prepared state $\tau_{x,\lambda}$, whose average over $\lambda$ yields the state $\rho_x$ accessible to the user. The state is sent to an uncharacterised measurement device controlled by $\lambda$. The goal is to certify randomness in the outcome $b$ against an adversary with access to $\lambda$.}
	\label{fig_QRNG}
\end{figure}

\subsection{Certified randomness}
Consider a preparation device that realises for the user the set of states $\mathcal{E}=\{\rho_x\}_x$, where the input $x$ is selected by the user. The internal operation of the device may depend on a hidden classical variable $\lambda$ known to the adversary. Conditioned on $\lambda$, input $x$ leads to the preparation of the state $\tau_{x,\lambda}$, with $\lambda$ distributed according to $q(\lambda)$. Consistency with the user's description requires that averaging over $\lambda$ reproduces $\rho_x$, as in Eq~\eqref{classicalmodel}. We call any collection $\{q(\lambda),\tau_{x,\lambda}\}_{x,\lambda}$ satisfying this condition a realisation of $\mathcal{E}$, and denote the set of all such realisations by $R(\mathcal{E})=\{(q,\tau):\text{Eq.~\eqref{classicalmodel} holds}\}$. The prepared states are sent to an uncharacterised measurement device whose input $y$ is also selected by the user. Its operation may depend on the  hidden variable $\lambda$, so that for input $y$ it implements the measurement $\{M_{b|y,\lambda}\}_b$. Conditioned on $\lambda$, the behaviour is therefore
\begin{equation}
	p(b|x,y,\lambda)=\tr\left(\tau_{x,\lambda}M_{b|y,\lambda}\right),
\end{equation}
whereas the user observes only
\begin{equation}\label{userbehaviour}
	p_{\mathrm{obs}}(b|x,y)	=\sum_\lambda q(\lambda)\tr\left(\tau_{x,\lambda}M_{b|y,\lambda}\right).
\end{equation}

Randomness is certified whenever the observed behaviour cannot be  predicted by an adversary with access to $\lambda$. This is quantified by the adversary's guessing probability. For a fixed realisation $(q,\tau)\in R(\mathcal{E})$ and fixed measurements $M$, once $x$ and $y$ are announced, the adversary optimally guesses the most likely outcome $b$. Averaging over the user's inputs  gives
\begin{equation}\label{Pguessstrat}
	g(q,\tau,M) 	\equiv 	\sum_{x,y}p_xp_y 	\sum_\lambda q(\lambda) \max_b 	\tr\left(\tau_{x,\lambda}M_{b|y,\lambda}\right),
\end{equation}
where $p_x$ and $p_y$ are arbitrary full-support distributions used to select $x$ and $y$. The adversary may choose any realisation $(q,\tau)\in R(\mathcal{E})$ and measurements $M$ compatible with the observed behaviour $p_{\mathrm{obs}}$. The optimal guessing probability is therefore
\begin{equation}\label{Pguess}
	G(\mathcal{E},p_{\mathrm{obs}}) \equiv 	\max_{\substack{(q,\tau,M) 	\\ \mathrm{compatible}(\mathcal{E},p_{\mathrm{obs}})}}	g(q,\tau,M).
\end{equation}
Consequently, the observed behaviour certifies randomness from the preparation device $\mathcal{E}$ if and only if $G(\mathcal{E},p_{\mathrm{obs}})<1$.

\subsection{Coherence characterises randomness}
We now establish the operational connection between coherence and randomness certification. Consider a user with a preparation device realising $\mathcal{E}$ subject to hidden classical control. We will show that coherence  determines whether the device is capable of certifying randomness. Specifically, we show that incoherent $\mathcal{E}$ can always be simulated in a way that is perfectly predictable to an adversary, whereas every coherent ensemble admits a prepare-and-measure experiment whose outcomes cannot be perfectly predicted.

We begin with the first implication. Suppose that $\mathcal{E}\in\mathcal{C}$. From Eq~\eqref{commutation}, we can write $\rho_x=\sum_{\lambda,k}q(\lambda) p(k|x,\lambda)E_{k|\lambda}$, where $\{E_{k|\lambda}\}_k$ is a basis. This immediately suggests a perfectly predictable implementation. The adversary preprograms the preparation device to sample $\lambda$ according to $q(\lambda)$, then samples $k$ according to $p(k|x,\lambda)$ and emits $E_{k|\lambda}$. Knowing $\lambda$, the measurement device can project onto the corresponding eigenbasis and recover $k$ deterministically. It then generates the final output $b$ according to $s(b|k,y,\lambda)=\tr\left(E_{k|\lambda}M_{b|y,\lambda}\right)$, which is also known to the adversary. Upon averaging over $\lambda$, this construction reproduces both the states $\rho_x$ and the behaviour in Eq~\eqref{userbehaviour} observed by the user, while every outcome remains predictable to the adversary. Hence, no randomness can be certified. A formal version of this argument is given in Appendix~\ref{AppRNG}.

The converse is both less immediate and more significant.  We show that if the outcomes of an experiment, for every possible choice of measurement, can be perfectly predicted from classical side information, the preparation ensemble must be incoherent, from which we conclude that any coherent ensemble generates randomness. Remarkably, this does not require tailoring a different type of certification task to each $\mathcal{E}$. A universal family of pairwise state-discrimination tasks is sufficient.

\begin{theorem}
	Let $\mathcal E=\{\rho_x\}_{x=1}^N$ be a set of states and $p_\mathrm{obs}$ a behaviour obtained from measuring these states. If $\mathcal{E}\notin \mathcal{C}$ there exists a $p_\mathrm{obs}$ such that $G(\mathcal E,p_\mathrm{obs})<1$.
\end{theorem}
\begin{proof}
	We focus on a specific input-output task first proposed in Ref~\cite{Brunner2013}. Let  the measurement device have an input consisting of a pair of numbers, $y=(i,j)$, such that $i<j$ and $i,j\in[N]$. It amounts to all  $\binom{N}{2}$ possible choices of two distinct labels in $[N]$. Conditioned on $x\in\{i,j\}$, the task is to discriminate whether $x=i$ or $x=j$. The former is given the answer $b=0$ and the latter $b=1$. We associate the figure of merit 
\begin{equation}\label{witness}
	W=\sum_{i<j} p(0|x=i,y=(i,j))-p(0|x=j,y=(i,j)).
\end{equation}
For a given  $(q,\tau)$, the largest possible value of $W$ is 
\begin{align}\nonumber\label{thmstep}
	W&=\sum_{i<j}\sum_\lambda q(\lambda) \tr\left(\left(\tau_{i,\lambda}-\tau_{j,\lambda}\right)M_{0|(i,j),\lambda}\right)\\
	&\leq \sum_{i<j}\sum_\lambda q(\lambda) D\left(\tau_{i,\lambda},\tau_{j,\lambda}\right),
\end{align} 
where $D(A,B)=\frac{1}{2}\|A-B\|_1$ is the trace-distance. The above inequality is saturated if and only if $M_{0|(i,j),\lambda}$, for every $(i,j,\lambda)$ with $q(\lambda)>0$, is the projector onto the positive eigenspace of $\tau_{i,\lambda}-\tau_{j,\lambda}$. 

Let us now define $W(\mathcal{E})$ as the largest value of $W$ compatible with the preparation set $\mathcal{E}$. This means that we maximise $W$ over all measurements and all $(q,\tau)\in R(\mathcal{E})$. We can then consider a specific realisation $(q,\tau)$ that, for the appropriately selected measurements, saturates the value $W(\mathcal{E})$. The resulting distribution, achieving the value $W(\mathcal{E})$, is denoted $p_\mathrm{obs}$.

Let us now consider the conditions under which $p_\mathrm{obs}$ certifies randomness. To that end, assume that it does not certify randomness, i.e.~that the adversary's guessing probability  is unit, namely $G(\mathcal{E},p_\mathrm{obs})=1$. It follows from Eq~\eqref{Pguessstrat} and \eqref{Pguess} that $ \max_b \tr\left(\tau_{x,\lambda}M_{b|(i,j),\lambda}\right)=1$ for every $(x,i,j,\lambda)$. Since $b$ is binary, this equivalently means $\nu_{x,ij,\lambda}\equiv \tr\left(\tau_{x,\lambda}M_{0|(i,j),\lambda}\right)\in\{0,1\}$. Consider now that we are given $(i,j,\lambda)$. In order to saturate Eq~\eqref{thmstep} we require 
\begin{equation}\label{constraintproof}
\Delta_{ij,\lambda}\equiv \nu_{i,ij,\lambda}-\nu_{j,ij,\lambda}=D(\tau_{i,\lambda},\tau_{j,\lambda}).
\end{equation}
Since $\nu_{x,ij,\lambda}\in\{0,1\}$, we have $\Delta_{ij,\lambda}\in\{-1,0,+1\}$, but $\Delta_{ij,\lambda}=-1$ is forbidden by Eq~\eqref{constraintproof}. Then, on the one hand, if  $\Delta_{ij,\lambda}=0$  the faithfulness of trace-distance implies $\tau_{i,\lambda}=\tau_{j,\lambda}$. On the other hand, if  $\Delta_{ij,\lambda}=1$ a unit trace-distance implies that $\tau_{i,\lambda}$ is orthogonal to $\tau_{j,\lambda}$. Hence, it follows that $[\tau_{i,\lambda},\tau_{j,\lambda}]=0$ $\forall i,j,\lambda$, which is the commutation property in Eq~\eqref{commutation}. This  implies    $\mathcal{E}\in\mathcal{C}$. Thus, we conclude that by choosing any $\mathcal{E}\notin\mathcal{C}$, the associated $p_\mathrm{obs}$ has  $G(\mathcal{E},p_\mathrm{obs})<1$.
 
\end{proof}

The construction in the proof gives a general recipe for randomness certification. The witness coefficients depend only on the number of preparations, and not on the particular states in $\mathcal{E}$. Thus, a single family of pair-discrimination tasks suffices to reveal certifiable randomness for any coherent $\mathcal{E}$. The dependence on $\mathcal{E}$ enters only through the attainable value $W(\mathcal{E})$ and the realisation achieving it.

Together, the two directions establish the desired operational equivalence. An incoherent preparation set can always be implemented so that all observed outcomes are  predictable when having access to the hidden classical information. In contrast, every coherent set admits an experiment for which such prediction is impossible. Coherence under classical control is therefore the necessary and sufficient resource for enabling  randomness certification against classical side information.

\section{Coherence certification}
The operational interpretation of coherence in terms randomness certification raises a natural question: how can we determine if $\mathcal{E}$ is coherent?
In this section, we develop SDP methods to certify quantum coherence. We first develop an SDP hierarchy that gives a sequence of outer approximations that converges to the set $\mathcal{C}$. This provides a necessary and sufficient criterion for coherence. Next, we develop a different SDP relaxation of $\mathcal{C}$, which serves as a practically useful tool for certifying coherence in state-of-the-art quantum preparation devices. We benchmark the performance of these methods in case studies. Then, we focus exclusively on qubit systems and show how dedicated convex programming methods enable highly accurate and scalable coherence certification. Finally, we adapt these tools to certify quantum channels; both when they support coherence and when they inevitably break it.

\subsection{Complete hierarchy}\label{Sec:complete}
We begin with constructing a converging hierarchy of outer SDP relaxations of $\mathcal{C}$. This means that we define a sequence of sets $\mathcal{C}_m$ for $m=2,\ldots,\infty$, each computable as an SDP, such that
\begin{equation}
\mathcal{C}_2 \supset \mathcal{C}_3 \supset \ldots \supset \mathcal{C}_\infty = \mathcal{C}.
\end{equation}

The conceptual idea behind the hierarchy
is the following. If the set $\mathcal{E}=\{\rho_{x}\}_{x=1}^{N}$
admits an incoherent model, namely  $\rho_{x}=\int d\lambda\ q(\lambda)\tau_{x,\lambda}$ where $[\tau_{x,\lambda},\tau_{x',\lambda}]=0$, 
one can use $m$ copies of the same preparation device to generate the multipartite states 
\begin{align}
O_{x_{1}\ldots x_{m}}=\int d\lambda\ q(\lambda)\tau_{x_{1},\lambda}\otimes\ldots\otimes\tau_{x_{m},\lambda},
\end{align}
where the tensor factors for different choices of the $x_i$ still commute for each $\lambda$. 

To see the implications of this idea, we first focus on the lowest relaxation level ($m=2$). For this, we define bipartite operators of the form
\begin{equation}\label{Operatordef}
O_{xy}=\int d\lambda \ q(\lambda) \tau_{x,\lambda}\otimes \tau_{y,\lambda}.
\end{equation}
These are trivially positive semidefinite. They are also separable, which we relax to a semidefinite condition: that they have a  positive partial-transpose. Moreover, their marginals are constrained because it follows from Eq~\eqref{classicalmodel} that 
\begin{align}\label{constraints1}\nonumber
& \tr_1(O_{xy})=\int d\lambda \ q(\lambda)  \tau_{y,\lambda} =\rho_y, \\
& \tr_2(O_{xy})=\int d\lambda \ q(\lambda) \tau_{x,\lambda}=\rho_x,
\end{align}
where $\tr_i$ denotes the partial trace over the $i$-th subsystem. The second of these constraints can be dropped if we establish the symmetry relation $O_{xy} = V_{(12)} O_{yx} V_{(12)}$, where $V_{(12)}=\sum_{ij}\ketbra{ij}{ji}$ is the swap operator. Note that it can be expressed as $V_{(12)}=d(\phi^+)^{T_1}$, where $\phi^+=\ketbra{\phi^+}$ and $\ket{\phi^+}=\frac{1}{\sqrt{d}}\sum_{i=0}^{d-1}\ket{ii}$ is the maximally entangled state. It follows that for any linear operators $X$ and $Y$,
\begin{equation}
\tr_1\left(V_{(12)} X_1\otimes Y_2\right)=XY.
\end{equation}
Applying this to the operators defined in Eq~\eqref{Operatordef}, we obtain the constraints
\begin{align}\label{constraints2}\nonumber
\tr_1\left(V_{(12)} O_{xy}\right)&=\int d\lambda \ q(\lambda) \tau_{x,\lambda}\tau_{y,\lambda}\\
&=\int d\lambda \ q(\lambda) \tau_{y,\lambda}\tau_{x,\lambda}=\tr_2\left(V_{(12)} O_{xy}\right),
\end{align}
where we have used the commutation relation in Eq~\eqref{commutation}. 

In summary, we have obtained the lowest level of the SDP relaxation hierarchy, namely the set $\mathcal{C}_2$. It is given by the following program. Given $\mathcal{E}=\{\rho_x\}_{x=1}^N$,
\begin{align}\nonumber\label{lvl1complete}
\text{find} \quad & \{O_{xy}\}_{x,y=1}^N \\\nonumber
\text{s.t} \quad 
& \tr_1\left(O_{xy}\right)=\rho_y, \quad\forall x,y \\\nonumber
& O_{xy}=V_{(12)} O_{yx} V_{(12)}\\\nonumber
& \tr_1\left(V_{(12)}O_{xy}\right) =\tr_2\left(V_{(12)}O_{xy}\right), \quad\forall x,y\\ \nonumber
& \left(O_{xy}\right)^{T_1}\succeq 0,\quad \forall x,y\\
& O_{xy} \succeq 0, \quad \forall x,y.
\end{align}
If this program is infeasible, we conclude that $\mathcal{E}\notin \mathcal{C}$ and hence the states are coherent.  If instead the SDP is feasible, we may try to certify its coherence using a higher relaxation level in the hierarchy.

To extend the construction of the lowest level to the $m$'th level of the hierarchy, we consider operators of the form
\begin{align}
O_{x_{1}\ldots x_{m}}=\int d\lambda\ q(\lambda)\tau_{x_{1},\lambda}\otimes\ldots\otimes\tau_{x_{m},\lambda},
\end{align}
for $x_1,\ldots,x_m\in\{1,\ldots,N\}$. Just as before, these operators are positive semidefinite, they have positive partial-transpose with respect to any subset of the $m$ systems, and by tracing out all but the $i$'th subsystem we recover the state $\rho_{x_i}$. For $m>2$ also additional constraints compatible with SDP are relevant. We give a full description of the hierarchy in Appendix~\ref{app:completehierarchy}. There, we also prove the key feature of this hierarchy, namely that it converges to the set $\mathcal{C}$ in the limit of large $m$ [see Theorem~\ref{thm:completeness}]. The main idea in this proof is based on establishing a correspondence between sets of incoherent states and quantum measurements that admit a decomposition in terms of measurements with commuting effects, which allows for a characterisation using de Finetti theorems for POVMs established in \cite{pusey2013quantum}. This result means that every coherent $\mathcal{E}$ is certified at some level of the SDP hierarchy. Conceptually, this reveals that coherence can be fully characterised with the  framework of semidefinite programming.

\subsection{Practical coherence criterion}\label{Sec:practical}
While the SDP relaxation hierarchy in the previous subsection provides a complete characterisation of coherence, it is not particularly efficient from a practical point of view. Therefore, we develop also a second  SDP method, which is incomplete but tailored to be practically more useful for certifying coherence in realistic sets of quantum states. The method draws on the SDP relaxation methodology based on block moment matrices, which was recently introduced in Ref~\cite{dalessandro2026}, and may be viewed as a generalisation of the Navascu\'es-Piornio-Ac\'in hierarchy \cite{Pironio2010}.

If $\mathcal{E}\in\mathcal{C}$, we can without loss of generality take  all the states $\tau_{x,\lambda}$ appearing in the decomposition \eqref{classicalmodel} to be pure. Say that we select one specific value of $\lambda$ and define the ordered list $L=\{\mathds{1}\}\cup\{\tau_{x,\lambda}\}_{x=1}^N$. Let $R\supset L$ be a set of monomials defined over $L$, i.e.~$R$ contains products of the elements appearing in $L$. For simplicity, let $R$ contain all products up to length  $m$. We now define the block moment matrix  
\begin{equation}
\Gamma^{(\lambda)}=\sum_{u,v\in R} \ketbra{i_u}{i_v}\otimes uv^\dagger,
\end{equation}
where $i_u$ ($i_v$) label the position of the monomial $u$ ($v$) in the list $R$. By construction, $\Gamma^{(\lambda)}$ is positive semidefinite. Let us now consider the  lowest relaxation level, namely $m=1$, for which $R=L$. The block moment matrix becomes 
\begin{equation}\label{BMM_lambda}
\Gamma^{(\lambda)}=\begin{pmatrix}
\mathds{1} & \tau_{1,\lambda} & \tau_{2,\lambda} & \ldots  & \tau_{N,\lambda}  \\
& \tau_{1,\lambda} & M^\lambda_{12} & \ldots & M^\lambda_{1N} \\
& & \tau_{2,\lambda} & \ldots & M^\lambda_{2N} \\
& & & \ddots & \vdots \\
& & & & \tau_{N,\lambda}
\end{pmatrix}.
\end{equation}
On the diagonal, we have used that the states are pure, namely $\tau_{x,\lambda}^2=\tau_{x,\lambda}$. Also, we  have introduced the optimisation variable  $M^\lambda_{ij}=\tau_{i,\lambda}\tau_{j,\lambda}$. In general, the product of two pure states is a non-Hermitian matrix. However, since $\mathcal{E}\in\mathcal{C}$ it follows that  the condition \eqref{commutation} is satisfied, i.e.~$[\tau_{i,\lambda},\tau_{j,\lambda}]=0$. The product of two commuting  states is a Hermitian positive semidefinite matrix. Hence, we have $M^\lambda_{ij}=(M^\lambda_{ij})^\dagger$ and $M^\lambda_{ij}\succeq 0$. Moreover, since  $\tau_{i,\lambda}(\mathds{1} - \tau_{j,\lambda}) \succeq 0$, we also have that $\tau_{i,\lambda} \succeq M_{ij}^{\lambda}$, and similarly $\tau_{j,\lambda} \succeq M_{ij}^{\lambda}$. We can now define a new block moment matrix as
\begin{equation}\label{BMM}
\Gamma=\int d\lambda \ q(\lambda) \Gamma^{(\lambda)} = \begin{pmatrix}
\mathds{1} & \rho_1 & \rho_2 & \ldots  & \rho_N  \\
& \rho_1 & M_{12} & \ldots & M_{1N} \\
& & \rho_2 & \ldots & M_{2N} \\
& & & \ddots & \vdots \\
& & & & \rho_N
\end{pmatrix},
\end{equation}
where $q(\lambda)$ is a probability density function. Via the decomposition in Eq~\eqref{classicalmodel} we recover the states $\rho_x$ as entries in $\Gamma$.  Furthermore, the properties identified for the matrix \eqref{BMM_lambda} carry over to $\Gamma$: we have $\Gamma\succeq 0$,  $M_{ij}=\int d\lambda\ q(\lambda) M^\lambda_{ij}$ being Hermitian and positive semidefinite, and $\rho_{i} \succeq M_{ij}$ and $\rho_j\succeq M_{ij}$. The existence of such a $\Gamma$ can be decided by SDP. Given $\mathcal{E}=\{\rho_x\}_{x=1}^N$, 
\begin{equation}
\begin{aligned}\label{lvl1practical}
\text{find} \quad & \Gamma \\
\text{s.t.} \quad 
& \Gamma_{1,1}=\mathds{1},\\
& \Gamma_{i+1,i+1}=\Gamma_{1,i+1}=\rho_{i}, \quad i=1,\dots, N, \\
& \Gamma_{i,i}\succeq \Gamma_{i,j}, \quad \Gamma_{j,j}\succeq \Gamma_{i,j}\succeq 0, \quad i\neq j \text{ and } i,j\geq 2,  \\
& \Gamma \succeq 0,
\end{aligned}
\end{equation}
where $\Gamma_{i,j}$ denotes the $d\times d$ matrix at block-row $i$ and block-column $j$. The infeasibility  of this SDP implies $\mathcal{E}\notin \mathcal{C}$ and hence it certifies the coherence of $\mathcal{E}$.


By examining the dual formulation of this SDP, it is also possible to extract witness-type criteria for coherence. These are inequalities of the type
	\begin{equation}
		W(\mathcal{E}) = \sum_{x=1}^{N} \tr(F_{x}\rho_{x}) \leq \beta^{\mathcal{C}},
	\end{equation}
	where $\lbrace F_{x}\rbrace_{x=1}^{N}$ are Hermitian operators and $\beta^{\mathcal{C}}$ is a bound satisfied by all incoherent $\mathcal{E}$. Therefore, a violation of this inequality implies $\mathcal{E} \notin \mathcal{C}$ (see Appendix~\ref{app_dual_lev1} for further details on the derivation). These criteria are highly relevant from an experimental point of view, as they allow us to deduce the coherence of $\mathcal{E}$ without having to rely on its tomographic reconstruction, a procedure that is  known to be notoriously resource-intensive.  Instead, they are based on probability estimates, in which the effects of fluctuations due to finite statistics can be easily addressed using standard methods based on for instance Chernoff's bound \cite{Chernoff1952Measure}.

The SDP criterion in Eq~\eqref{lvl1practical} can also be viewed as the lowest level in an SDP hierarchy. The extension to higher levels is achieved  by  adding higher-order monomials to the list $R$. However, the hierarchy does not converge to $\mathcal{C}$ in the limit of including all monomials of arbitrary high level in $R$. For our purposes here, considering higher relaxation levels is of secondary interest, because increasing the level is associated with significantly increased computational cost. Our goal here is to introduce a general SDP criterion which is efficient and useful, and for that it needs to perform well already at the lowest level. In the next section, we  evidence that the criterion in Eq~\eqref{lvl1practical} has these merits.

\subsection{Comparison of methods and benchmarking of performance}
Practical coherence certification methods should aspire to meet three qualitative goals simultaneously. These are 
\begin{enumerate}
	\headeritem{Generality:} to apply to any set of states.
	
	\headeritem{Scalability:} to be computable for sizable $N$ and $d$.
	
	\headeritem{Accuracy:} to provide useful estimates of the set $\mathcal{C}$.
\end{enumerate}
The two methods, presented in sections~\ref{Sec:complete} and \ref{Sec:practical} respectively, both meet the first goal in the sense that they can in principle be used to analyse any set of quantum states $\mathcal{E}$. However, as expected, both of them exhibit a trade-off between scalability and accuracy. For scalability, the most relevant objective is to practically enable the evaluation of the SDP criteria for  sets $\mathcal{E}$ with at least several tens of states and/or high Hilbert space dimension. To make this possible, it is natural to focus on the lowest relaxation level for each of the two SDP methods. 

In this regard, the second method has a distinct advantage: the criterion in Eq~\eqref{lvl1complete} features $N^2$ different SDP matrices, each of size $d^2$, together with $\order{N^2}$ semidefinite and linear constraints of the same dimension. Counting the scalar degrees of freedom, this results in a program scaling as $\order{d^4N^2}$ in the number of variables and constraints. By comparison, the criterion in Eq~\eqref{lvl1practical} is based on an SDP matrix of size $d(N+1)$, containing $\order{N^2}$ hermitian variables of size $d$ and featuring $\order{N^2}$ linear and semidefinite constraints. Thus the dominant scaling in both variables and constraints is reduced from $\order{d^4N^2}$ to $\order{d^2N^2}$. Therefore, one expects the second SDP method to be more successful in terms of the second goal above. However, this does not address the third goal. As we will show,  when comparing the lowest level of the two methods, the second one is not only significantly less computationally costly but also significantly more accurate. This provides the justification for its practicality. 

%

We begin with showcasing the scalability of the second SDP method at the lowest level defined in Eq~\eqref{lvl1practical}. To this end, we have verified that it can easily deal with a large number of qubit states. For example, we have selected $N=100$ random qubit states and in less than eight minutes of computation on a standard desktop computer with $32$~GB of RAM, obtained a certificate of coherence. Similarly, we have selected $N=70$ random qutrit states and computed a certificate of  coherence in less than 11 minutes.  While these examples probe the many-state regime, we can also probe the large-dimension regime. For instance, consider the minimal example of just two  states: we select them as noisy eigenstates of the computational and Fourier bases, namely $\rho_1=v\ketbra{0}+\frac{1-v}{d}\mathds{1}$ and $\rho_2=v\ketbra{+}+\frac{1-v}{d}\mathds{1}$, where $\ket{+}=\frac{1}{\sqrt{d}}\sum_{i=0}^{d-1}\ket{i}$. We use the SDP in Eq~\eqref{lvl1practical} to compute a bound on $v\in[0,1]$ above which $\mathcal{E}_v=\{\rho_1,\rho_2\}$ is certified to be coherent. For example, we can compute the case of $d=150$ in 12 minutes and certify coherence for $v\gtrsim 0.9246$.

Next, we consider two case studies in which we compare the first and second SDP method.  To this end, consider first a set of states corresponding to all basis elements of both the computational basis and the Fourier basis, both subjected to isotropic noise.  Specifically, we consider $\mathcal{E}_{\text{bases}}=\{\rho_{x_1,x_2}\}$ where $x_1=1,2$ is the basis index and $x_2=0,\ldots,d-1$ is the index of the basis element. Thus, $N=2d$. We select $\rho_{1,x_2}=v\ketbra{x_2}+\frac{1-v}{d}\mathds{1}$ and $\rho_{2,x_2}=v\ketbra{e_{x_2}}+\frac{1-v}{d}\mathds{1}$, where $\ket{e_j}=\frac{1}{\sqrt{d}}\sum_{k}e^{\frac{2\pi i}{d}jk}\ket{k}$ is the Fourier basis. Using both SDP methods, we compute bounds on the isotropic noise parameter $v$ and illustrate the results in Fig~\ref{Fig:plot_benchmark}(blue). We have allowed up to $2.5$ hours of computation per point, which is why the first method truncates at $d=10$ while the second method truncates at $d=13$. We observe that the second method significantly outperforms the first in terms of accuracy: for one method the visibility threshold increases with $d$ whereas for the other it decreases.

\begin{figure}[t!]
	\centering
	\includegraphics[width=1\linewidth]{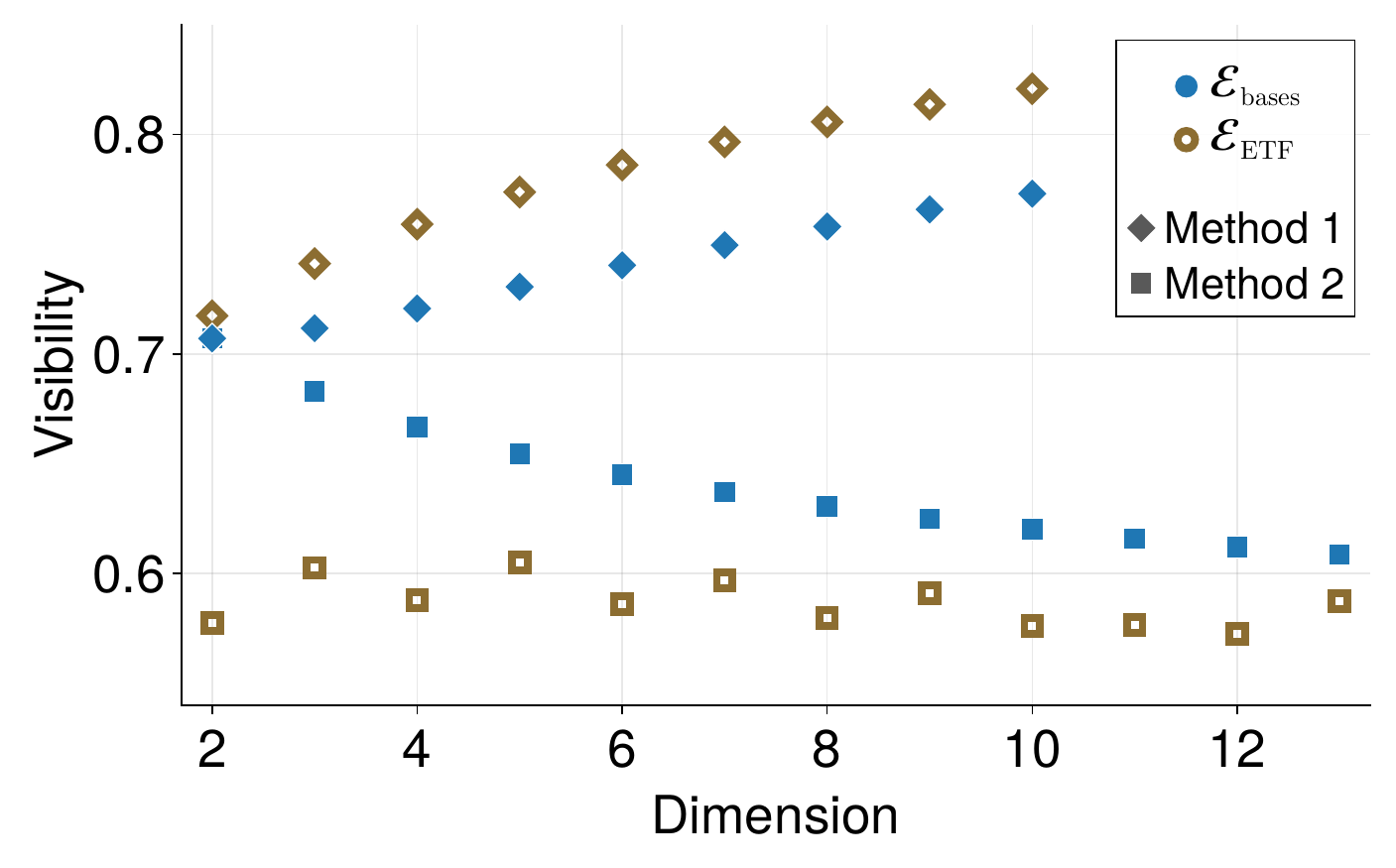}
	\caption{Bounds on the critical isotropic noise required to certify coherence for the two sets $\mathcal{E}_\text{bases}$ and $\mathcal{E}_\text{ETF}$. The points were obtained using the methods introduced in section~\ref{Sec:complete} (diamonds) and in section~\ref{Sec:practical} (squares). The  results from the second method are not monotone for  $\mathcal{E}_\text{ETF}$ because  different families are used to construct the equiangular tight frame for different $d$.}
	\label{Fig:plot_benchmark}
\end{figure}

Qualitatively similar results are encountered also in the second case study. We choose a set  $\mathcal{E}_\text{ETF}$ comprised of  pure states subject to isotropic noise, where the pure states form a so-called equiangular tight frame. This means that every pair of states is separated by the same angle and that this angle is as large as it possibly can be. Specifically, one has $\lvert\braket{\psi_j}{\psi_k} \rvert^2=c$ for any $j\neq k$, where $c=\frac{N-d}{d(N-1)}$ is the largest possible constant. Such states have been used in dimension witnessing \cite{Brunner2013}, tests of non-projective measurements \cite{Rosset2019} and are closely related to the well-known symmetric informationally complete POVMs \cite{Renes2004}. In our case study, we consider $N=2d$ states and use systematic constructions of  equiangular tight frames  provided in Ref~\cite{iverson2024}; see Appendix~\ref{App:ETF} for details. The results, displayed in  Fig~\ref{Fig:plot_benchmark}, again show that at the lowest level, the second SDP method can address larger values of $d$ while giving significantly more accurate bounds.

\subsection{Qubit systems}\label{Sec:qubit}
In the analysis of coherence, qubit systems ($d=2$) is the most important special case. Naturally, the methods outlined in sections~\ref{Sec:complete} and \ref{Sec:practical} also apply for qubits. However, it is possible to design improved coherence certification methods by tailoring them specifically for qubits. To achieve this, we combine two independent results in the literature: one that allows qubit coherence to be expressed as a joint measurability problem, and another that shows how to efficiently bound jointly measurable qubit observables. 

To this end, consider that we are given a set of $N$ qubits, $\mathcal{E}=\{\rho_x\}_{x=1}^N$. To this set we associate a set of unbiased dichotomic measurements, $\mathcal{M}=\{\rho_x,\mathds{1}-\rho_x\}_{x=1}^N$. For simplicity, we can write $\rho_x=\frac{1}{2}\left(\mathds{1}+\vec{n}_x\cdot \vec{\sigma}\right)$, where $\vec{n}_x$ is the Bloch vector of $\rho_x$. The set of measurements can then be associated with antipodal pairs of Bloch vectors $\{\vec{n}_x,-\vec{n}_x\}_{x=1}^N$. It was proven in Ref~\cite{Cobucci2026} that $\mathcal{E}\in \mathcal{C}$ if and only if $\mathcal{M}$ is jointly measurable. Joint measurability means that we can write 
\begin{equation}\label{JMdef}
\rho_x=\sum_\mu p(\mu)p(0|x,\mu)G_\mu,
\end{equation} 
where $\{G_\mu\}_{\mu}$ is a POVM and $\{p(0|x,\mu),p(1|x,\mu)\}$ is a conditional probability distribution. Note that the second outcome-operator in the set $\mathcal{M}$, namely $\mathds{1}-\rho_x$, is fixed by normalisation. Whether a decomposition \eqref{JMdef} is possible can be decided by an SDP \cite{Guhne2023}, but the computation is not scalable since the number of variables  grows exponentially in $N$.

However,  efficient and accurate inner and outer approximations of the set of jointly measurable qubit observables, $\mathcal{M}$, have been proposed in Ref~\cite{Porto2025}. Combined with the above, the feasibility of an inner approximation can be interpreted as an incoherent model for $\mathcal{E}$. Analogously, the infeasibility of an outer approximation can be interpreted as a certificate of coherence for $\mathcal{E}$. To use the method, one first exploits that  $\mathcal{M}$ is jointly measurable if and only if the assemblage $\{\sigma_{a|x}\}$, obtained from applying the $x$'th measurement in $\mathcal{M}$ to half of the maximally entangled state $\phi^+$ and registering the outcome $a$, admits a local hidden state model  \cite{Quintino2014, Uola2014}. A local hidden state model is a decomposition of the form \cite{Wiseman2007}
\begin{align}\nonumber\label{LHSmodel}
& \sigma_{0|x}\equiv \tr_A\left(\rho_x\otimes \mathds{1}\phi^+\right)=\sum_\mu p(\mu) p(0|x,\mu)\sigma_\mu, \\
& \sigma_{1|x}\equiv\tr_A\left((\mathds{1}-\rho_x)\otimes \mathds{1}\phi^+\right)=\sum_\mu p(\mu) p(1|x,\mu)\sigma_\mu,
\end{align}
where $\sigma_\mu$ is a state. 

To build an inner approximation to joint measurability, we can  choose $\{\sigma_\mu\}$ as some finite set of states. On the Bloch sphere, this can be visualised as  a polyhedron whose vertices are points on the sphere's surface.  Evaluating the existence of such a restricted local hidden state model is a linear program. This can be seen by defining $\tilde{p}(a|x,\mu)=p(\mu) p(a|x,\mu)$ and using as variables in the program $\tilde{p}(a|x,\mu)$ and $p(\mu)$. The relevant constraints are  $\tilde{p}(a|x,\mu)\geq 0$ and $\sum_a \tilde{p}(a|x,\mu)=p(\mu)$ $\forall x$. In contrast, to build an outer approximation of joint measurability, we  let $\sigma_\mu$ be a set of arbitrary Hermitian matrices with the property that their Bloch vectors form the vertices of a polyhedron that contains the Bloch sphere. This can be achieved by appropriately elongating the Bloch vectors of the polyhedron used for the inner approximation. This outer approximation can, in analogy with the inner approximation, be computed as a linear program. How accurately the inner and outer approximations correspond to joint measurability is determined by how well the employed polyhedron approximates the Bloch sphere; see Ref~\cite{Porto2025} for further details.

As a compromise between accuracy of the results and the efficiency of computing them, we choose to approximate the Bloch sphere with a polyhedron corresponding to $220$ pure qubit states. It is obtained from an edgewise subdivision of the Bloch sphere (following Appendix C of Ref~\cite{liu2025}); we provide the vertices in Ref~\cite{githubRep}. We use the same polyhedron for both inner and outer approximation; the former is related to the latter by scaling the Bloch vector lengths by a factor of $r\approx 0.9934$. The fact that this number is nearly unit indicates high accuracy, i.e.~that the inner and outer approximations will nearly match. As we will demonstrate in the next section, this method allows both coherence and incoherence to be certified in systems with a large number of qubit states.

\subsection{Coherence-breaking channels}\label{Sec:CBC}
Lastly, we consider whether coherence  can survive when passed through a given quantum channel, i.e.~a completely positive trace-preserving map $\Lambda:\mathcal{D}(\mathbb{C}^d)\rightarrow \mathcal{D}(\mathbb{C}^{d'})$, where we label by $\mathcal{D}(\mathbb{C}^s)$ the set of density matrices over an $s$-dimensional Hilbert space. To this end, we  define the image of $\Lambda$ as  the set $\mathcal{E}_\Lambda\equiv \{\Lambda(\rho) \quad \text{s.t.} \quad \rho \in \mathcal{D}(\mathbb{C}^d)\}$. We define $\Lambda$ as coherence-breaking (CB) if its image satisfies $\mathcal{E}_\Lambda\in\mathcal{C}$. This means that if we were to send all quantum states through the channel, the resulting set of states  would admit an incoherent model. Without loss of generality, one can simplify this definition to consider only pure states, $\psi$, namely 
\begin{equation}
\Lambda \text{ is coherence-breaking} \quad \Leftrightarrow\quad \{\Lambda(\psi)\}_\psi\in\mathcal{C}\quad,
\end{equation}
where $\psi$ runs over all pure states in dimension $d$.  

An elementary result is that there is no strict relation between CB channels and entanglement-breaking channels. One direction is trivial; consider a qubit channel that measures $\sigma_Z$ and prepares $\ket{0}$ or $\frac{\ket{0}+\ket{1}}{\sqrt{2}}$ depending on the outcome. Evidently, it is entanglement-breaking. However it is not CB because it maps the states $\{\ket{0},\ket{1}\}$ into $\{\ket{0},\frac{\ket{0}+\ket{1}}{\sqrt{2}}\}$ which are coherent since they are pure and non-commuting. The channel has simply generated the coherence itself. The other direction is exemplified by the depolarisation channel, $\Lambda(\rho)=v\rho+\frac{1-v}{d}\mathds{1} $, which is known to be entanglement-breaking if and only if $v\leq \frac{1}{d-1}$. This number is smaller than $\frac{H_d-1}{d-1}$, which in \cite{Cobucci2026} was implicitly shown to be the  threshold at which the  channel becomes CB. However, for more general channels, proving the CB property in the affirmative requires more general tools. In Appendix~\ref{AppCB} we deploy our SDP relaxation methods both to prove that a channel is CB and that it is can preserve cohernece.

\section{Conclusions}
Coherence is an indispensable feature of quantum theory and it has therefore been intensively studied from both foundational, resource theoretic and quantum technology perspectives. Here, we have investigated coherence when the quantum preparation device is not fully controlled, but instead subjected to classical side information. Our main conceptual result provides an operational interpretation of such quantum coherence, showing that it is the necessary and sufficient resource for whether a preparation device is useful for quantum random number generation in  prepare-and-measure experiments. This establishes a one-to-one connection to this quantum technology and gives further motivation to the question of how one can certify that a given device indeed is coherent.

To address this question, we have developed a toolbox based on convex programming. Firstly, we have introduced a converging hierarchy of semidefinite programs for bounding the set of incoherent quantum states. This is conceptually important, as it proves that coherence can be completely characterised in terms of semidefinite programming methods. Secondly, we have adopted the practical point of view, in which we are presented with a state-of-the-art quantum preparation device, capable of generating perhaps dozens of states in dozens of dimensions, and tasked to prove its coherence. For this purpose, we have developed coherence certification methods that are generally applicable, computationally efficient and yield useful results in a series of case studies. Thirdly, for qubit preparation devices, we leveraged both conceptual and methodological connections to the topic of joint measurability to enable  accurate characterisation of coherence that is scalable to well over one thousand states. All these techniques were then adapted to also analyse coherence-breaking channels. This toolbox provides the methods  needed for benchmarking coherence.

\begin{acknowledgments}
This work is financially supported by the Swedish Foundation for Strategic Research, the Knut and Alice Wallenberg Foundation through the Wallenberg Center for Quantum Technology (WACQT), the Crafoord Foundation, the Krapperup Foundation  and  the Swedish Research Council under Contract No.~2023-03498. R.B.~acknowledges support by Deutsche Forschungsgemeinschaft (DFG,
German Research Foundation) under Germany’s Excellence Strategy – Cluster of Excellence Matter and Light for Quantum Computing (ML4Q) EXC 2004/1 – 390534769. N.W.~acknowledges support by EIN Quantum NRW.
\end{acknowledgments}

\bibliography{bibliography}

\appendix
\onecolumngrid

\section{Incoherent states cannot certify randomness}\label{AppRNG}

In this appendix, we show that no incoherent preparation device can be used to generate randomness in the prepare and measure scenario. For this, assume that the ensemble is incoherent, $\mathcal{E}\in\mathcal{C}$. By Eq~\eqref{commutation}, follows that the states admit the decomposition $\rho_x=\sum_{\lambda,k}q(\lambda) p(k|x,\lambda)E_{k|\lambda}=\sum_\lambda q(\lambda)\rho_{x,\lambda}$, where $\{E_{k|\lambda}\}_k$ defines a family of mutually orthogonal rank-1 projectors.  When the measurements implemented are $\{M_{b|y,\lambda}\}_b$, the correlations observed by the user are 
\begin{equation}
	p_{\mathrm{obs}}(b|x,y) = \sum_\lambda q(\lambda)\mathrm{tr}\left(\rho_{x,\lambda}M_{b|y,\lambda}\right) = \sum_{\lambda, k}q(\lambda)p(k|x,\lambda)s(b|k,y,\lambda),
\end{equation}
where $s(b|k,y,\lambda)=\mathrm{tr}\left(E_{k|\lambda}M_{b|y,\lambda}\right)$. We now construct an alternative realisation of these correlations in which the adversary can predict every outcome, thereby showing the impossibility to certify randomness.

For this, after sampling $\lambda$ from $q(\lambda)$, let the adversary preprogram the preparation device by sampling, for every possible input $x$, a label $k_x$ according to $p(k_x|x,\lambda)$. Similarly, for every pair $(k,y)$, the adversary samples an outcome $b_{k,y}$ from the distribution $s(b_{k,y}|k,y,\lambda)=\tr\left(E_{k|\lambda}M_{b|y,\lambda}\right)$. For the input pair $(x,y)$, the predetermined measurement outcome is then $b_{k_x,y}$. We collect all these choices into a new shared classical variable $\mu\equiv \left(\lambda,\{k_x\}_x,\{b_{k,y}\}_{k,y}\right)$, distributed according to
\begin{equation}
	p(\mu)=q(\lambda)\prod_{x} p(k_x|x,\lambda)\prod_{k,y} s(b_{k,y}|k,y,\lambda).
\end{equation}
Importantly, $p(\mu)$ is independent of the user's choices of $x$ and $y$, it rather represents a predetermined table of responses for all possible inputs. When the user chooses $x$ and the classical control is $\mu$, the preparation device emits $\sigma_{x,\mu}=E_{k_x|\lambda}$. The measurement projects $\sigma_{x,\mu}$ in its eigenbasis, namely $N_{b|y,\mu}=\sum_k E_{k|\lambda}\delta_{b_{k,y},b}$. Orthogonality of the projectors gives $\mathrm{tr}\left(\sigma_{x,\mu}N_{b|y,\mu}\right)
=\delta_{b,b_{k_x,y}}$, hence the outcome is deterministic conditioned on $\mu$ and therefore known to the adversary. This has no randomness, while both reproducing the user's state and observed behaviour $p_\mathrm{obs}$. To show this, averaging over $\mu$ we obtain
\begin{equation}
	\begin{split}
	\sum_\mu p(\mu)\sigma_{x,\mu}&=\sum_{\lambda, k_u, b_{j,l}}q(\lambda)\prod_{u}p(k_u|u,\lambda)\prod_{j,l}s(b_{j,l}|j,l,\lambda)E_{k_x|\lambda} \\
	&=\sum_{\lambda,k_x}q(\lambda)p(k_x|x,\lambda)E_{k_x|\lambda}\times\prod_{u\neq x}\left[\sum_{k_u}p(k_u|u,\lambda)\right]\prod_{j,l}\left[\sum_{b_{j,l}}s(b_{j,l}|j,l,\lambda)\right]\\
	&=\sum_{\lambda,k}q(\lambda)p(k|x,\lambda)E_{k|\lambda}=\rho_x
	\end{split}
\end{equation}
and similarly
\begin{equation}
	\begin{split}
		\sum_\mu p(\mu) \tr\left(\sigma_{x,\mu} N_{b|y,\mu}\right)&=\sum_{\lambda, k_u, b_{j,l}}q(\lambda)\prod_{u}p(k_u|u,\lambda)\prod_{j,l}s(b_{j,l}|j,l,\lambda)\delta_{b_{k_x,y},b}\\
		&=\sum_{\lambda,k_x}q(\lambda)p(k_x|x,\lambda) s(b|k_x, y, \lambda)\times\prod_{u\neq x}\left[\sum_{k_u}p(k_u|u,\lambda)\right]\prod_{j,l\neq k_x, y}\left[\sum_{b_{j,l}}s(b_{j,l}|j,l,\lambda)\right]\\
		&=\sum_{\lambda,k}q(\lambda)p(k|x,\lambda)s(b|k,y,\lambda)=p(b|x,y).\\
	\end{split}
\end{equation}

\section{Complete hierarchy for coherent sets} \label{app:completehierarchy}
We give a detailed outlined of the SDP relaxation hierarchy discussed in section~\ref{Sec:complete} and prove that it converges to the incoherent set $\mathcal{C}$.

As noted in the main text, the idea behind the hierarchy
is that any incoherent model for $\mathcal{E}=\{\rho_{x}\}_{x=1}^{N}$, i.e.~$\rho_{x}=\int d\lambda\ q(\lambda)\tau_{x,\lambda}$ where $[\tau_{x,\lambda},\tau_{x',\lambda}]=0$, 
can be used to consider the output of $m$ copies of the same preparation device, giving access to the multipartite states 
\begin{align}
O_{x_{1}\ldots x_{m}}=\int d\lambda\ q(\lambda)\tau_{x_{1},\lambda}\otimes\ldots\otimes\tau_{x_{m},\lambda},\label{eq:O1tom}
\end{align}
where the tensor factors for different choices of the $x_i$ still commute for each $\lambda$. The goal is to formulate
the existence of the $O_{x_{1}\ldots x_{m}}$ as a semidefinite program.
To that end, we infer semidefinite constraints from the existence
of Eq~\eqref{eq:O1tom}. We get for all $x_{j}\in\{1,\ldots,N\}$:
\begin{enumerate}
\item[(c1)] $O_{x_{1}\ldots x_{m}}\succeq 0$.
\item[(c2)] $\tr_{m}(O_{x_{1}\ldots x_{m}})=O_{x_{1}\ldots x_{m-1}}$ and $O_{x_{1}}=\rho_{x_{1}}$,
where $\tr_{m}$ denotes the partial trace over the last system.
\item[(c3)] Let $S_m$ be the permutation group of $m$ elements and let $\pi\in S_{m}$. Its  unitary representation is given by
\begin{align}
V_{\pi}=\sum_{i_{1}\ldots i_{m}=1}^{N}\ketbra{\pi(i_{1}\ldots i_{m})}{i_{1}\ldots i_{m}},
\end{align}
where $\pi(i_{1}\ldots i_{m})=i_{\pi^{-1}(1)}\ldots i_{\pi^{-1}(m)}$.
Then the symmetry of Eq.~(\ref{eq:O1tom}) implies
\begin{align}
V_{\pi}O_{x_{1}\ldots x_{m}}V_{\pi}^{\dagger}=O_{\pi(x_{1}\ldots x_{m})}
\end{align}
for all $\pi$.
\item[(c4)] Consider the representation of a cyclic permutation of the first $l$ systems (called an $l$-cycle) $V_{(12\ldots l)}^{\dagger}$.
Then a straightforward calculation reveals $\tr_{12\ldots l-1}[V_{(12\ldots l)}^{\dagger}A_{1}\otimes\ldots\otimes A_{l}]=A_{1}A_{2}\ldots A_{l}$.
Now, the tensor factors of the $O_{x_{1}\ldots x_{m}}$ all commute,
implying that for all $2\leq l\leq m$ and all permutations $\pi\in S_{l}$,
\begin{align}
\tr_{1,2\ldots l-1}[V_{(12\ldots l)}^{\dagger}O_{x_{1}\ldots x_{l}x_{l+1}\ldots x_{m}}]=\tr_{1,2\ldots l-1}[V_{(12\ldots l)}^{\dagger}O_{\pi(x_{1}\ldots x_{l})x_{l+1}\ldots x_{m}}].
\end{align}
\item[(c5)] Finally, the $O_{x_{1}\ldots x_{m}}$ constitute fully separable
quantum states, implying that they have positive partial transposition
for each bipartition,  
\begin{align}
O_{x_{1}\ldots x_{m}}^{T_{X}}\succeq 0,
\end{align}
for all $X\subset\{1,\ldots,m\}$, and $T_{X}$ denote the partial
transposition w.r.t.~the systems in $X$.
\end{enumerate}
Collecting these observations, we define the $m$'th level of our SDP
hierarchy as 
\begin{align}
\text{find~} & \{O_{x_{1}\ldots x_{m}}\}_{x_{1}\ldots x_{m}=1}^{N}\label{eq:Ohierarchy}\\
\text{s.t.~} & \text{(c1) to (c5) hold.}\nonumber 
\end{align}
Note that this can easily be turned into an SDP yielding
upper bounds on $v_{\text{worst}}$ and $v^{*}$ defined in Eq.~\eqref{robustness_iso}. For instance,
for $v^{*}$, we minimise over a real parameter $v$ instead
and change (c2) to demand $O_{x_{1}}=(1-v)\rho_{x_{1}}+\frac{v}{d}\openone$.

For clarity, the hierarchy for $m=2$ reads 
\begin{align*}
\text{find~} & \{O_{xy}\}_{x,y=1}^{N}\\
\text{s.t.~} & \forall x,y=1\ldots N: &  & O_{xy}\succeq 0,\\
 &  &  & \tr_{2}O_{xy}=\rho_{x},\\
 &  &  & V_{(12)}O_{xy}V_{(12)}=O_{yx},\\
 &  &  & \tr_{1}[V_{(12)}O_{xy}]=\tr_{1}[V_{(12)}O_{yx}],\\
 &  &  & O_{xy}^{T_{1}}\succeq 0.
\end{align*}

\begin{theorem}\label{thm:completeness}
The SDP hierarchy in Eq~\eqref{eq:Ohierarchy} is complete, i.e., a
set $\mathcal{E}=\{\rho_{x}\}_{x=1}^{N}$ is incoherent if and only
if the $O_{x_{1}\ldots x_{m}}$ satisfying constraints (c1) to (c5)
exist for all $m$.
\end{theorem}

\begin{proof}
We first build an $(N+1)$-effect POVM $\tilde{\mathcal{E}}=\{E_{x}\}_{x=1}^{N+1}$ based on the states in $\mathcal{E}$;
\begin{equation}
E_{x}=\begin{cases}
\rho_{x}/N & x\in\{1,\ldots,N\},\\
\openone-\sum_{y=1}^{N}\rho_{y}/N & x=N+1.
\end{cases}\label{eq:POVMeffects}
\end{equation}
Note that the renormalisation is necessary to ensure positivity of
the last effect. Now, if $\mathcal{E}$ is incoherent,  we can decompose $\tilde{\mathcal{E}}$ in terms of POVMs with
commuting effects via
\begin{equation}
E_{x}=\int d\lambda\ q(\lambda)E_{x|\lambda}\label{eq:POVMdecomposition}
\end{equation}
where 
\begin{align}
E_{x|\lambda}=\begin{cases}
\tau_{x,\lambda}/N & x\in\{1,\ldots,N\},\\
\openone-\sum_{y=1}^{N}\tau_{x,\lambda}/N & x=N+1,
\end{cases}
\end{align}
with $[E_{x|\lambda},E_{y|\lambda}]=0$, $\sum_{x=1}^{N+1}E_{x|\lambda}=\openone$.
We will now show the following:
\begin{enumerate}
\item There exists a complete hierarchy to decide the existence of a decomposition
(\ref{eq:POVMdecomposition}) for a given POVM. We denote the corresponding
solution of level $m$ by $\{\tilde{O}_{x_{1}\ldots x_{m}}\}_{x_{1}\ldots x_{m}=1}^{N+1}$.
\item From the solution of the POVM hierarchy, we can extract a decomposition
of $\mathcal{E}$.
\item From a solution $\{O_{x_{1}\ldots x_{m}}\}_{x_{1}\ldots x_{m}=1}^{N}$
of the $m$'th level of the hierarchy in Eq~\eqref{eq:Ohierarchy},
we can construct a solution $\{\tilde{O}_{x_{1}\ldots x_{m}}\}_{x_{1}\ldots x_{m}=1}^{N+1}$
for the $m$'th level of the POVM hierarchy.
\end{enumerate}
Together, the three results imply the claim, as existence of the $O_{x_{1}\ldots x_{m}}$
for all $m$ imply via 3.~the existence of corresponding $\tilde{O}_{x_{1}\ldots x_{m}}$,
which implies via 1.~the existence of a decomposition of the corresponding
POVM, from which a decomposition can be read off via 2.

\bigskip

\emph{Proof of 1.} For a given POVM $\tilde{\mathcal{E}}=\{E_{x}\}_{x=1}^{N+1}$,
consider the hierarchy
\begin{align}
\text{find} & \{\tilde{O}_{x_{1}\ldots x_{m}}\}_{x_{1}\ldots x_{m}=1}^{N+1}\label{eq:Otildehierarchy}\\
\text{s.t.} & \forall x_{1}\ldots x_{m}=1\ldots N+1: &  & \tilde{O}_{x_{1}\ldots x_{m}}\succeq 0,\label{eq:Ot_pos}\\
 &  &  & \sum_{y=1}^{N+1}\tilde{O}_{x_{1}\ldots x_{m-1}y}=\tilde{O}_{x_{1}\ldots x_{m-1}}\otimes\openone,\label{eq:Ot_sum}\\
 &  &  & \tilde{O}_{x_{1}}=E_{x_{1}},\label{eq:Ot_E}\\
 &  &  & \forall\pi\in S_{m}:\,V_{\pi}\tilde{O}_{x_{1}\ldots x_{m}}V_{\pi}^{\dagger}=\tilde{O}_{\pi(x_{1}\ldots x_{m})},\label{eq:Ot_sym}\\
 &  &  & \forall l\in\{2,\ldots m\},\pi\in S_{l}:\,\nonumber \\
 &  &  & \quad\tr_{1,2\ldots l-1}[V_{(1,2\ldots l)}^{\dagger}\tilde{O}_{x_{1}\ldots x_{l}x_{l+1}\ldots x_{m}}]\nonumber \\
 &  &  & \hfill=\tr_{1,2\ldots l-1}[V_{(1,2\ldots l)}^{\dagger}\tilde{O}_{\pi(x_{1}\ldots x_{l})x_{l+1}\ldots x_{m}}],\label{eq:Ot_swap}\\
 &  &  & \tr(\tilde{O}_{x_{1}\ldots x_{m}})=\prod_{j=1}^{m}\tr(E_{x_{j}}).\label{eq:Ot_tr}
\end{align}
If this hierarchy is feasible for every level $m$, then constraints~(\ref{eq:Ot_pos})-(\ref{eq:Ot_sym})
imply via Theorem~5.4 in Ref~\cite{pusey2013quantum} that there exist
POVMs $\{E_{x|\lambda}\}_{x=1}^{N+1}$ and a probability distribution
$q(\lambda)$ such that for all $m$ and $x_{1}\ldots x_{m}\in\{1,\ldots,N+1\}$:
\begin{equation}
\tilde{O}_{x_{1}\ldots x_{m}}=\int d\lambda\ q(\lambda)E_{x_{1}|\lambda}\otimes\ldots\otimes E_{x_{m}|\lambda}.\label{eq:Ot_solution}
\end{equation}
Next, we show that almost all $\{E_{x|\lambda}\}_{x=1}^{N+1}$ have
commuting effects. To that end, consider level $m=4$, fix $x_{1}=x_{2}=x$,
$x_{3}=x_{4}=y$ for some $x,y\in\{1,\ldots,N+1\}$. Using constraint~(\ref{eq:Ot_swap}),
we get
\begin{align*}
0 & =\tr_{1,2,3}[V_{(1234)}^{\dagger}(\tilde{O}_{xyyx}-\tilde{O}_{xyxy}-\tilde{O}_{yxyx}+\tilde{O}_{yxxy})]\\
 & =\int d\lambda\ q(\lambda)(E_{x|\lambda}E_{y|\lambda}E_{y|\lambda}E_{x|\lambda}-E_{x|\lambda}E_{y|\lambda}E_{x|\lambda}E_{y|\lambda}\\
 & \phantom{=\int d\lambda\ q(\lambda)(}-E_{y|\lambda}E_{x|\lambda}E_{y|\lambda}E_{x|\lambda}+E_{y|\lambda}E_{x|\lambda}E_{x|\lambda}E_{y|\lambda}\\
 & =\int d\lambda\ q(\lambda)[E_{x|\lambda},E_{y|\lambda}][E_{x|\lambda},E_{y|\lambda}]^{\dagger}.
\end{align*}
As $[E_{x|\lambda},E_{y|\lambda}][E_{x|\lambda},E_{y|\lambda}]^{\dagger}$
is positive semidefinite, this implies that $[E_{x|\lambda},E_{y|\lambda}]=0$
for almost all $\lambda$, i.e., the effects commute.

Finally, we impose constraint~(\ref{eq:Ot_tr}) for fixed $m$ and
$x=x_{1}=x_{2}=\ldots=x_{m}$. Inserting Eq.~(\ref{eq:Ot_solution})
yields 
\begin{align}
\int d\lambda\ q(\lambda)\tr(E_{x|\lambda})^{m}=\tr(E_{x})^{m}
\end{align}
for all $x$ and $m$, implying that $\tr(E_{x|\lambda})=\tr(E_{x})$
for almost all $\lambda$. In conclusion, we have shown that if hierarchy~(\ref{eq:Otildehierarchy})
exists for a given POVM $\tilde{\mathcal{E}}=\{E_{x}\}_{x=1}^{N+1}$
for all $m$, then there exists a decomposition of $\tilde{\mathcal{E}}$
in terms of POVMs $\{E_{x|\lambda}\}_{x=1}^{N+1}$ with commuting
effects and $\tr(E_{x|\lambda})=\tr(E_{x})$. 

\bigskip

\emph{Proof of 2.} For $\mathcal{E}=\{\rho_{x}\}_{x=1}^{N}$, consider
the $(N+1)$-effect POVM $\tilde{\mathcal{E}}$ with effects given
in Eq.~(\ref{eq:POVMeffects}) and assume that hierarchy~(\ref{eq:Otildehierarchy})
is feasible for each $m$. By 1., it follows that we can write for
all $x\in\{1,\ldots N\}$
\begin{align*}
\rho_{x}/N & =\int d\lambda\ q(\lambda)E_{x|\lambda}
\end{align*}
with $[E_{x|\lambda},E_{y|\lambda}]=0$ and $\tr(E_{x|\lambda})=1/N$.
Multiplying both sides by $N$ yields a decomposition of $\mathcal{E}$
in terms of commuting sets of states.

\bigskip

\emph{Proof of 3.} To complete the proof, we have to show that feasibility
of~(\ref{eq:Ohierarchy}) implies feasibility of~(\ref{eq:Otildehierarchy})
for the POVM with effects~(\ref{eq:POVMeffects}). To that end, consider
the solution $\{O_{x_{1}\ldots x_{m}}\}_{x_{1}\ldots x_{m}=1}^{N}$
of the $m$'th level of~(\ref{eq:Ohierarchy}), and we set
\begin{align}
\tilde{O}_{x_{1}\ldots x_{m}}=\frac{1}{N^{m}}O_{x_{1}\ldots x_{m}}
\end{align}
for all $x_{1}\ldots x_{m}\in\{1,\ldots,N\}.$ These automatically
fulfill constraints (\ref{eq:Ot_pos}), (\ref{eq:Ot_E}), (\ref{eq:Ot_sym}),
(\ref{eq:Ot_swap}) and (\ref{eq:Ot_tr}). It remains to construct
the $\tilde{O}_{x_{1}\ldots x_{m-1},N+1}$, $\tilde{O}_{x_{1}\ldots x_{m-2},N+1,N+1}$
and so on (and the remaining ones follows from symmetry). First, we set 
\begin{align*}
\tilde{O}_{x_{1}\ldots x_{m-1},N+1} & =\tilde{O}_{x_{1}\ldots x_{m-1}}\otimes\openone-\sum_{y=1}^{N}\tilde{O}_{x_{1}\ldots x_{m-1},y}\\
 & =\frac{1}{N^{m-1}}O_{x_{1}\ldots x_{m-1}}\otimes\openone-\frac{1}{N^{m}}\sum_{y=1}^{N}O_{x_{1}\ldots x_{m-1},y}
\end{align*}
such that the sum constraint~(\ref{eq:Ot_sum}) is fulfilled. The
remaining constraints are fulfilled due to their linearity. However,
positivity of $\tilde{O}_{x_{1}\ldots x_{m-1},N+1}$ has to be ensured.
To that end, we write
\begin{align*}
N^{m}\tilde{O}_{x_{1}\ldots x_{m-1},N+1} & =NO_{x_{1}\ldots x_{m-1}}\otimes\openone-\sum_{y=1}^{N}O_{x_{1}\ldots x_{m-1},y}\\
 & =\sum_{y=1}^{N}(O_{x_{1}\ldots x_{m-1}}\otimes\openone-O_{x_{1}\ldots x_{m-1},y})\\
 & =\sum_{y=1}^{N}[\openone^{(1\ldots m-1)}\otimes R^{(m)}](O_{x_{1}\ldots x_{m-1},y}),
\end{align*}
where $\openone^{(1\ldots m-1)}\otimes R^{(m)}$ denotes the partial
reduction map applied to system $m$ with $R(\rho)=\tr(\rho)\openone-\rho$ \cite{horodecki1999reduction}. The reduction map is known to be
a positive, but not completely positive map which is strictly weaker
than PPT, meaning that if partial transposition applied to system
$m$ yields a positive semidefinite matrix (as ensured by constraint
(c5) of hierarchy~(\ref{eq:Ohierarchy})), applying the partial reduction
map on system $m$ must yield a positive semidefinite matrix as well,
ensuring that $\tilde{O}_{x_{1}\ldots x_{m-1},N+1}\succeq 0$, as well.

Finally, we have to consider analogously the cases with multiple indices
set to $N+1$. For instance, we set (with the systems the operators
act upon written in superscript)
\begin{align*}
\tilde{O}_{x_{1}\ldots x_{m-2},N+1,N+1} & =\tilde{O}_{x_{1}\ldots x_{m-2}}^{(1,\ldots,m-2)}\otimes\openone^{(m-1)}\otimes\openone^{(m)}-\sum_{y=1}^{N}\tilde{O}_{x_{1}\ldots x_{m-2},y}^{(1,\ldots,m-1)}\otimes\openone^{(m)}\\
 & \phantom{=}-\sum_{z=1}^{N}\tilde{O}_{x_{1}\ldots x_{m-2},z}^{(1,\ldots,m-2,m)}\otimes\openone^{(m-1)}+\sum_{y,z=1}^{N}\tilde{O}_{x_{1}\ldots x_{m-2},y,z}^{(1,\ldots,m)}\\
 & \propto\sum_{y,z=1}^{N}[\openone^{(1\ldots m-2)}\otimes R^{(m-1)}\otimes R^{(m)}](O_{x_{1}\ldots x_{m-2},y,z}),
\end{align*}
which is positive since $O_{x_{1}\ldots x_{m-2},y,z}^{T_{\{m-1,m\}}}\succeq 0$.
The remaining cases follow analogously, establishing 3.~and thereby
the claim.
\end{proof}

\section{Coherence witness criteria from dual SDP} \label{app_dual_lev1}

In this appendix, we derive the dual SDP of the coherence block matrix SDP in \eqref{lvl1practical} and show how to extract witness-type criteria for certifying coherence. For concreteness, consider that we want to quantify the coherence of a set of states by applying to it isotropic noise, i.e. $\mathcal{E}_v = \lbrace \rho_{x}^{(v)} =  v\rho_x + \frac{1-v}{d} \openone\rbrace_x$, and determine bounds on $v$ above which $\mathcal{E}_v$ is coherent. The primal problem reads as:
\begin{equation}\label{primal_block_lev1}
	\begin{aligned}
		v^{*} = \max_{\lbrace v, \Gamma_{i,j} \rbrace} & \quad v \\
		\text{s.t.}&\quad \sum_{i=1}^{N+1} \ketbra{i}{i} \otimes \Gamma_{i,i} +  \sum_{i=1}^{N} \sum_{j=i+1}^{N+1}(\ketbra{i}{j} \otimes \Gamma_{i,j} + \ketbra{j}{i} \otimes \Gamma_{i,j})\succeq 0,\\
		& \quad \Gamma_{1,j} = v\rho_{j-1} + \frac{1-v}{d}\openone, \quad \forall j = 2,\dots, N+1,\\
		& \quad \Gamma_{j,j} = \Gamma_{1,j}, \quad \forall j = 2,\dots, N+1,\\
		& \quad \Gamma_{i,j} \succeq 0, \quad \forall i = 2,\dots,N; j = i+1,\dots,N+1,\\
		& \quad \Gamma_{i,i} \succeq \Gamma_{i,j}, \quad \forall i = 2,\dots,N; j = i+1,\dots,N+1,\\
		& \quad \Gamma_{j,j} \succeq \Gamma_{i,j}, \quad \forall i = 2,\dots,N; j = i+1,\dots,N+1,\\
		& \quad \Gamma_{1,1} = \openone.
	\end{aligned}
\end{equation}
By introducing a Lagrangian multiplier for each constraint
\begin{equation}
	\begin{aligned}
		&\sum_{i=1}^{N+1} \ketbra{i}{i} \otimes \Gamma_{i,i} +  \sum_{i=1}^{N} \sum_{j=i+1}^{N+1}(\ketbra{i}{j} \otimes \Gamma_{i,j} + \ketbra{j}{i} \otimes \Gamma_{i,j})\succeq 0 \quad \rightarrow \quad Z \succeq 0,\\
		&\Gamma_{1,j} = v\rho_{j-1} + \frac{1-v}{d}\openone \quad \rightarrow \quad Y_{j}, \quad \forall j = 2,\dots,N+1,\\
		&\Gamma_{j,j} = \Gamma_{1,j} \quad \rightarrow \quad B_{j}, \quad \forall j = 2,\dots,N+1,\\
		&\Gamma_{i,j} \succeq 0 \quad \rightarrow \quad W_{ij} \succeq 0, \quad \forall i = 2,\dots,N; j = i+1,\dots,N+1,\\
		&\Gamma_{i,i} \succeq \Gamma_{i,j} \quad \rightarrow \quad \gamma_{ij} \succeq 0, \quad \forall i = 2,\dots,N; j = i+1,\dots,N+1,\\
		&\Gamma_{j,j} \succeq \Gamma_{i,j} \quad \rightarrow \quad \theta_{ij} \succeq 0, \quad \forall i = 2,\dots,N; j = i+1,\dots,N+1,\\
		&\Gamma_{1,1} = \openone \quad \rightarrow \quad R,
	\end{aligned}
\end{equation}
we can write the Lagrangian of the problem as
\begin{equation}\label{lagrangian_primal}
	\begin{aligned}
		\mathcal{L} = & v + \sum_{i=1}^{N+1}\tr[Z \ketbra{i}{i} \otimes \Gamma_{i,i}] + \sum_{i=1}^{N} \sum_{j=i+1}^{N+1}\tr\left[Z(\ketbra{i}{j} \otimes \Gamma_{i,j} + \ketbra{j}{i} \otimes \Gamma_{i,j})\right] + \sum_{j=2}^{N+1}\tr\left[Y_{j}\left(\Gamma_{1,j} - v\rho_{j-1} -\frac{1-v}{d}\openone\right)\right]\\
		& \, + \sum_{j=2}^{N+1}\Tr\left[B_{j}(\Gamma_{j,j} - \Gamma_{1,j})\right] + \sum_{i=2}^{N}\sum_{j=i+1}^{N+1}\tr[W_{ij} \Gamma_{i,j}] + \sum_{i=2}^{N}\sum_{j=i+1}^{N+1}\tr\left[\gamma_{ij}(\Gamma_{i,i} - \Gamma_{i,j})\right]\\
		& \, + \sum_{i=2}^{N}\sum_{j=i+1}^{N+1}\tr\left[\theta_{ij}(\Gamma_{j,j} - \Gamma_{i,j})\right] + \tr\left[R(\Gamma_{1,1} - \openone)\right].
	\end{aligned}
\end{equation}
The Lagrangian multipliers will become the variables of the dual problem. In order to ensure that the Lagrangian is never smaller than the optimal value of the primal problem, i.e. $\mathcal{L} \geq v^{*}$, we must restrict to $Z \succeq 0$, $W_{ij} \succeq 0$, $\gamma_{ij} \geq 0$ and $\theta_{ij} \geq 0$. Indeed, in this case the terms involving $Z$, $W_{ij}$, $\gamma_{ij}$ and $\theta_{ij}$ will always be non-negative for any feasible solution of the primal. Using that $\tr\left[Z(\ketbra{i}{j} \otimes G_{ij})\right] = \tr\left[G_{ij} \Tr_{A}(Z(\ketbra{i}{j} \otimes \openone)\right] = \tr\left[G_{ij}Z_{ji}\right]$, we can factorize the primal variables in \eqref{lagrangian_primal} as
\begin{equation}
	\begin{aligned}
		\mathcal{L} &= v\left[1 - \sum_{j=2}^{N+1}\tr\left[Y_{j}\left(\rho_{j-1} - \frac{\openone}{d}\right)\right]\right] + \tr\left[\Gamma_{1,1}(Z_{11} + R)\right] + \tr\left[\Gamma_{2,2}\left(Z_{22} + B_{2} + \sum_{j=3}^{N+1}\gamma_{2j}\right)\right]\\
		&\, + \sum_{j=3}^{N}\tr\left[\Gamma_{j,j}\left(Z_{jj} + B_{j} + \sum_{i=j+1}^{N+1}\gamma_{ji} + \sum_{i=2}^{j-1}\theta_{ij}\right)\right] +  \tr[\Gamma_{N+1,N+1}\left(Z_{N+1N+1} + B_{N+1} + \sum_{i=2}^{N}\theta_{iN+1}\right)]\\
		& \, + \sum_{j=2}^{N+1}\Tr\left[\Gamma_{1,j}(Z_{1j} + Z_{j1} + Y_{j} - B_{j})\right] + \sum_{i=2}^{N}\sum_{j=i+1}^{N+1}\tr\left[\Gamma_{i,j}(Z_{ij} + Z_{ji} + W_{ij} - \gamma_{ij} - \theta_{ij})\right]\\
		&\, - \tr[R] - \sum_{j=2}^{N+1}\tr\left[Y_{j}\frac{\openone}{d}\right].
	\end{aligned}
\end{equation}
By imposing the following contraints on the dual variables,
\begin{equation}
	\begin{aligned}
		& 1 - \sum_{j=2}^{N+1}\tr\left[Y_{j}\left(\rho_{j-1} - \frac{\openone}{d}\right)\right] = 0\\
		& Z_{11} + R = 0, \quad Z_{22} + B_{2} + \sum_{j=3}^{N+1}\gamma_{2j} = 0, \quad Z_{jj} + B_{j} + \sum_{i=j+1}^{N+1} \gamma_{ji} +  \sum_{i=2}^{j-1}\theta_{ij} = 0, \quad \forall j = 3,\dots,N,\\
		& Z_{N+1N+1} + B_{N+1} + \sum_{i=2}^{N}\theta_{iN+1} = 0,\\
		& Z_{1j} + Z_{j1} + Y_{j} - B_{j} = 0, \quad \forall j = 2,\dots,N+1,\\
		& Z_{ij} + Z_{ji} + W_{ij} - \gamma_{ij} - \theta_{ij} = 0, \quad \forall i = 2,\dots,N; j = i+1,\dots,N+1,
	\end{aligned}
\end{equation}
the Lagrangian does not present any dependence on the dual variables. Therefore, the dual problem reads as
\begin{equation}
	\begin{aligned}
		\min_{\lbrace Z_{ij},W_{ij},Y_{j},B_{j},\gamma_{ij},\theta_{ij} \rbrace} & \quad  - \tr[R] - \sum_{j=2}^{N+1}\tr\left[Y_{j}\frac{\openone}{d}\right]\\
		\text{s.t.}&\quad \sum_{i=1}^{N+1} \ketbra{i}{i} \otimes Z_{ii} +  \sum_{i=1}^{N} \sum_{j=i+1}^{N+1}(\ketbra{i}{j} \otimes Z_{ij} + \ketbra{j}{i} \otimes Z_{ji})\succeq 0,\\
		&\quad  1 - \sum_{j=2}^{N+1}\tr\left[Y_{j}\left(\rho_{j-1} - \frac{\openone}{d}\right)\right] = 0,\\
		& \quad Z_{11} + R = 0, \quad Z_{22} + B_{2} + \sum_{j=3}^{N+1}\gamma_{2j} = 0,\\
		& \quad Z_{jj} + B_{j} + \sum_{i=j+1}^{N+1} \gamma_{ji} +  \sum_{i=2}^{j-1}\theta_{ij} = 0, \quad \forall j = 3,\dots,N,\\
		& \quad Z_{N+1N+1} + B_{N+1} + \sum_{i=2}^{N}\theta_{iN+1} = 0,\\
		& \quad Z_{1j} + Z_{j1} + Y_{j} - B_{j} = 0, \quad \forall j = 2,\dots,N+1,\\
		& \quad Z_{ij} + Z_{ji} + W_{ij} - \gamma_{ij} - \theta_{ij} = 0, \quad \forall i = 2,\dots,N; j = i+1,\dots,N+1,\\
		& \quad W_{ij} \succeq 0, \quad \gamma_{ij} \succeq 0, \quad \theta_{ij} \succeq 0, \quad \forall i = 2,\dots,N; j = i+1,\dots,N+1.
	\end{aligned}
\end{equation}
Note that $R$, $B_{j}$, $Y_{j}$ and $W_{ij}$ are slack variables, since they do not appear in the objective function. Therefore, we can solve for them using the constraints. The simplified dual reads as
\begin{equation}\label{dual_block_lev1}
	\begin{aligned}
		\min_{\lbrace Z_{ij}, \gamma_{ij}, \theta_{ij} \rbrace} & \quad   1 + \tr[Z_{11}] + \sum_{j=2}^{N+1}\tr[\beta_{j}\rho_{j-1}],\\
		\text{s.t.}&\quad \sum_{i=1}^{N+1} \ketbra{i}{i} \otimes Z_{ii} +  \sum_{i=1}^{N+1} \sum_{j=i+1}^{N+1}(\ketbra{i}{j} \otimes Z_{ij} + \ketbra{j}{i} \otimes Z_{ji})\succeq 0,\\
		& \quad M_{j} = Z_{1j} + Z_{j1} + Z_{jj}, \quad \forall j = 2,\dots,N+1,\\
		& \quad \beta_{2} = M_{2} + \sum_{j=3}^{N+1}\gamma_{2j}, \quad \beta_{N+1} = M_{N+1} + \sum_{i=2}^{N}\theta_{iN+1}, \quad \beta_{j} = M_{j} + \sum_{i=j+1}^{N+1}\gamma_{ji} + \sum_{i=2}^{j-1}\theta_{ij}, \quad \forall j = 2,\dots,N,\\
		& \quad 1 + \sum_{j=2}^{N+1}\tr\left[\beta_j\left(\rho_{j-1}-\frac{\openone}{d}\right)\right]= 0,\\
		&\quad  Z_{ij} + Z_{ji} - \gamma_{ij} - \theta_{ij} \preceq 0, \quad \gamma_{ij} \succeq 0, \quad \theta_{ij} \succeq 0, \quad \forall i = 2,\dots,N; j = i+1,\dots,N+1.
	\end{aligned}
\end{equation}
In order to prove that strong duality holds for this problem, i.e. the solution of the primal and the dual problems coincide, we need to check that Slater's conditions are fulfilled for \eqref{primal_block_lev1} \cite{Skrzypczyk2023}. In particular, we need that (i) it is possible to find a strictly feasible solution to the primal, (ii) the optimal objective is bounded. The second condition is always satisfied from the positivity of the states. To prove (i), we can consider $\Gamma_{1,j} = \frac{\mathds{1}}{d}$ for all $j=2,\dots,N+1$ and $\Gamma_{i,j} = \frac{\mathds{1}}{d^2}$ for all $i=2,\dots,N$ and $j=i+1,\dots,N+1$, which always give a solution to the primal in terms of strictly positive variables. Now, we can show that the feasible solutions $\lbrace Z_{ij}, \gamma_{ij} \rbrace_{i,j}$ of the dual problem \eqref{dual_block_lev1} provide a witness for coherence detection. Indeed, given a set $\mathcal{E}_v = \lbrace \rho_j^{(v)} = v\rho_j + \frac{1-v}{d}\mathds{1}\rbrace_j$, we can define
\begin{equation}
	W(\mathcal{E}_{v}) =  \tr[Z_{11}] + \sum_{j=2}^{N+1}\tr[\beta_{j}\rho_{j-1}^{(v)}].
\end{equation}
From strong duality, we know that $v^{*} = 1 + \tr[Z_{11}] + \sum_{j=2}^{N+1}\tr[\beta_{j}\rho_{j-1}]$ for the optimal values of $\lbrace Z_{ij}, \gamma_{ij} \rbrace_{ij}$. Therefore, for any set $\mathcal{E}_v$ with $v \leq v^{*}$, we get
\begin{equation}
	\begin{aligned}
		W(\mathcal{E}_v) &= \tr[Z_{11}] + v\underbrace{\sum_{j=2}^{N+1}\tr[\beta_{j}\rho_{j-1}]}_{=v^{*}-1-\tr[Z_{11}]} + (1-v)\underbrace{\sum_{j=2}^{N+1}\tr\left[\beta_j\frac{\openone}{d}\right]}_{=v^{*}-\tr[Z_{11}]}\\
		& = \tr[Z_{11}] + v(v^{*}-1-\tr[Z_{11}]) + (1-v)(v^{*}-\tr[Z_{11}]) = v^{*}-v \geq 0.
	\end{aligned}
\end{equation}
This inequality is satisfied by all $\mathcal{E}_{v} \in \mathcal{C}$. Therefore, finding $W(\mathcal{E}) \ngeq 0$ implies that $\mathcal{E} \notin \mathcal{C}$.
	
\section{Equiangular tight frame construction}\label{App:ETF}

In this appendix, we show how to explicitly construct families of $2d$ equiangular tight frames in dimension $d$ for all dimensions $d=2,3,\dots, 13$, following the techniques developed in \cite{ETF2d2024}. 

\subsection{Definitions and notation}

Let $\{\psi_i\}_{i=1}^n$ be pure states in a $d$-dimensional Hilbert space and consider $\Psi=[\psi_1, \dots,\psi_n]$ to be the matrix whose $i$-th column is $\psi_i$\\ \\
\textbf{Definition 1.} The set $\{\psi_i\}_{i=1}^n \subset \mathbb C^d$ is called a $d\times n$ \emph{equiangular tight frame} (ETF) if it is both tight and equiangular. Tightness means that $\Psi\Psi^\ast = \frac{n}{d} \mathds{1}_d$, while equiangularity means that there exists a constant $\alpha>0$ such that $\left|\left<\psi_i,\psi_j\right>\right|=\alpha$ for all $i\neq j$.\\ \\
\textbf{Definition 2.} The \emph{Gram matrix} of the frame $\{\psi_i\}_{i=1}^n$ is the matrix $G=\Psi^\dagger \Psi$ whose $(i,j)$-th entry is the inner product $\langle\psi_i,\psi_j\rangle$. For an ETF, the Gram matrix can be written as $G = I_n + \alpha S$. Here $S$ is a Hermitian matrix with zero diagonal and off-diagonal entries with unit modulus. The matrix $S$ is called \emph{signature matrix} of the ETF.\\ \\
\textbf{Definition 3.} A \emph{conference matrix} of order $n$ is an $n\times n$ matrix $C$ with diagonal entries equal to $0$, off-diagonal entries equal to $\pm1$, and satisfying $C^TC=(n-1)\mathds{1}_n$. A conference matrix is called \emph{symmetric} if $C^T=C$, and \emph{skew-symmetric} if $C^T=-C$. For convenience of notation, we call the first $C_\text{sym}$ and the second $C_\text{skew}$.\\ \\
\textbf{Definition 4.} Let $A$ be the adjacency matrix of a graph on $n$ vertices, and let $J_n$ denote the $n\times n$ all-ones matrix. The matrix $Q=J_n-\mathds{1}_n-2A$ is called the \emph{Seidel adjacency matrix} of the graph. Similarly, if $A$ is the adjacency matrix of an orientation of a complete graph (digraph), the matrix $T=A-A^\top$ is called its \emph{skew-Seidel adjacency matrix}. \\ \\
\textbf{Definition 5.} Let $\mathbb F_q$ be the finite field with $q$ elements, where $q$ is an odd prime power. If $q\equiv 1 \pmod 4$, the \emph{Paley graph} on $\mathbb F_q$ is the graph whose vertices are the elements of $\mathbb F_q$, where two distinct vertices $x,y\in\mathbb F_q$ are adjacent if and only if $x-y$ is a nonzero square in $\mathbb F_q$. If $q\equiv 3 \pmod 4$, the same condition defines the \emph{Paley digraph}, where there is an oriented edge from $x$ to $y$ if and only if $x-y$ is a nonzero square in $\mathbb F_q$. 

\subsection{Paley equiangular tight frames}

The first family of ETFs is constructed from conference matrices. It follows from \textbf{Definition 3} that any symmetric conference matrix has eigenvalues in $\pm \sqrt{n-1}$. Since $\tr(C_\text{sym})=0$, each eigenvalue has multiplicity $\frac{n}{2}$. Hence, $\mathds{1}_n+\frac{1}{\sqrt{n-1}}C_\text{sym}$ is the Gramm matrix of a real $\frac{n}{2}\times n$ ETF. By a similar argument $\mathds{1}+\frac{i}{\sqrt{n-1}}C_\text{skew}$ is the Gram matrix of a complex $\frac{n}{2}\times n$ ETF for a skew-Symmetric conference matrix. Furthermore, a conference graph, and similarly a digraph, can be obtained directly from \textbf{Definition 4} by extending appropriately the corresponding Seidel or skew-Seidel matrix as
\begin{equation}\label{eq:seidel_to_conference}
	C_\text{sym}=\begin{bmatrix}
		0 & \textbf{1}^T\\
		\textbf{1} & S
	\end{bmatrix} \qquad\text{and}\qquad	
	C_\text{skew}=\begin{bmatrix}
		0 & \textbf{1}^T\\
		-\textbf{1} & T
	\end{bmatrix}.
\end{equation}
For every odd prime power $q$, the Paley graph for $q\equiv 1 \pmod 4$, or the Paley digraph for $q\equiv 3 \pmod 4$, gives then rise to an ETF as follows. Starting from its adjacency matrix $A$, one forms the Seidel matrix $Q=J_q-I_q-2A$ or the skew-Seidel matrix $T=A-A^T$, then constructs the corresponding conference matrix via Eq.~\eqref{eq:seidel_to_conference}, and finally obtains the Gram matrix
\begin{equation}
	G=I_{q+1}+\frac{1}{\sqrt q}\,C_{\mathrm{sym}}
	\qquad\text{or}\qquad
	G=I_{q+1}+\frac{i}{\sqrt q}\,C_{\mathrm{skew}}.
\end{equation}
This yields an ETF of $q+1$ vectors in dimension $\frac{q+1}{2}$, real when $q\equiv 1 \pmod 4$ and complex when $q\equiv 3 \pmod 4$.

\subsection{ETF through doubling}

The previous example covers all cases in the main text except for $d=8$ and $d=11$. To obtain ETFs of size $2d$ in these remaining dimensions, we use the doubling construction described in \cite{ETF2d2024}. More precisely, we consider that if $S$ is the signature matrix of a $d\times n$ ETF with $n=2d+k$ and $k\in\{-1,0,1\}$ it is possible to take $c=(n-2d)\sqrt{\frac{n-1}{d(n-d)}}$ and $\beta=-c+\varepsilon i\sqrt{1-c^2}$, for some $\varepsilon\in\{-1,1\}$ such that 
\begin{equation}\label{eq:doubling_ETF}
	\widetilde S=
	\begin{bmatrix}
		S & S+\beta I_n\\
		S+\overline{\beta}I_n & -S
	\end{bmatrix}
\end{equation}
is the signature matrix of a complex $n\times 2n$ ETF. For $d=8$, we construct the Paley ETF of size $4\times 8$, obtained from the Paley digraph with $q=7$ through \textbf{Definition 5}, \textbf{Definition 4}, and Eq.~\eqref{eq:seidel_to_conference}. Applying Eq.~\eqref{eq:doubling_ETF} and retrieving the corresponding Gram matrix then yields a complex $8\times 16$ ETF. 

For $d=11$, the Paley construction is not available. In this case we instead use the Renes-Strohmer construction \cite{RENES2007}. Let $T$ be the skew-Seidel matrix of the Paley digraph on $q$ vertices with $q\equiv 3 \pmod{4}$, the matrix 
\begin{equation}
	G_{\mathrm{RS}}=\frac{1}{q+1}(J_q+q\mathds{1}_q+i\sqrt q\,T)
\end{equation}
is the Gram matrix of a complex ETF of $q$ vectors in dimension $\frac{q+1}{2}$. Since here the common angle is $\alpha=\frac{1}{\sqrt{q+1}}$, the corresponding signature matrix is
\begin{equation}\label{eq:renes_strohmer}
	S_{\mathrm{RS}}=\sqrt{q+1}\,(G_{\mathrm{RS}}-\mathds{1}_q)
	=\frac{1}{\sqrt{q+1}}\bigl(J_q-\mathds{1}_q+i\sqrt q\,T\bigr).
\end{equation}
To obtain the $11\times22$ ETF we only need to choose $q=11$ and use Eq.~\eqref{eq:doubling_ETF} applied to the result of Eq.~\eqref{eq:renes_strohmer}.

\section{Certification of coherence-breaking channels}\label{AppCB}
In this appendix, we show how to deploy the techniques introduced in the main text to certify whether a given channel, $\Lambda$, preserves or breaks coherence. Since
\begin{equation}
	\Lambda \text{ is coherence-breaking} \quad \Leftrightarrow\quad \{\Lambda(\psi)\}_\psi\in\mathcal{C},
\end{equation}
where $\psi$ runs over all pure states in dimension $d$, coherence certification of $\Lambda$ requires to identify a set of pure states $\mathcal{E}=\{\psi_x\}_{x}$ such that $\{\Lambda(\psi_x)\}_{x}$ is coherent. The simplest approach is to guess a set $\{\psi_x\}_{x}$ and evaluate the appropriate SDP relaxation with the goal of certifying the coherence of $\{\Lambda(\psi_x)\}_{x}$. However, a more reliable approach is to choose the number of states, $N$, and iteratively search for the optimal choice of 
$\mathcal{E}=\{\psi_x\}_{x}$ for maximising the coherence of $\{\Lambda(\psi_x)\}$. For concreteness, let us quantify coherence via the isotropic noise model, i.e. defining $v^{*}(\Lambda(\mathcal{E}))$ as the largest $v$ such that $\big\{v\Lambda(\rho_x)+\frac{1-v}{d}\openone \big\}_x \in\mathcal{C}$. For each possible choice of $\mathcal{E}$, the appropriate SDP relaxation provides a bound $\bar{v}(\Lambda(\mathcal{E})) \geq v^*(\Lambda(\mathcal{E}))$ and via its dual (see Appendix \ref{app_dual_lev1}) it returns a set of operators $\lbrace \beta_x\rbrace_{x=1}^{N}$ from which a coherence witness $W(\Lambda(\mathcal{E})) = \sum_{x=1}^{N} \tr(\beta_x\Lambda(\psi_x)) + \textit{const}$ can be extracted. The search over the state space is performed using the largest violation of the witness as an oracle. The states for which this occurs can be found via the Choi-Jamio\l{}kowski representation of the channel, $\eta_{A'A} = (\Lambda \otimes \openone)[\phi^{+}]$, where $\ket{\phi^{+}} = \frac{1}{\sqrt{d}}\sum_{k=0}^{d-1}\ket{kk}$ \cite{Choi1975,Jamiolkowski1972}. Indeed, since $\Lambda(\psi_{x}) = d \tr_{A}((\openone_{A'} \otimes \psi_{x}^{T}) \, \eta_{A'A})$, we get
\begin{equation}
	\begin{aligned}
		W(\Lambda(\mathcal{E})) &= \sum_{x=1}^{N}\tr(\beta_{x}\Lambda(\psi_{x})) + const. = d\sum_{x=1}^{N}\tr(\beta_{x} \tr_{A}((\openone_{A'} \otimes \psi_{x}^{T}) \, \eta_{A'A})) + const.\\
		&= d\sum_{x=1}^{N}\tr((\beta_{x}\otimes \openone_{A}) (\openone_{A'} \otimes \psi_{x}^{T}) \, \eta_{A'A}) + const.\\
		&= d\sum_{x=1}^{N}\tr(\psi_{x}^{T}\tr_{A'}((\beta_{x}\otimes \openone_{A}) \, \eta_{A'A})) + const.\\
		&= d\sum_{x=1}^{N}\tr(\psi_{x}\tr_{A'}((\beta_{x}\otimes \openone_{A}) \, \eta_{A'A})^{T}) + const.\\
	\end{aligned}
\end{equation}

Therefore, the states which leads to the largest violation of the witness are those that correspond to the smallest eigenvalues of the operators $\lbrace d\big(\tr_{A'}\big((\beta_x \otimes \openone)\eta_{A'A}\big)\big)^{T}\rbrace_{x=1}^{N}$.
These states are used as the input of the SDP hierarchy in the next iteration of the algorithm.
The whole process is repeated until desired convergence is found, returning an estimate for the set of states $\lbrace \psi_{x} \rbrace_{x=1}^{N}$ which maximises the coherence of the states emerging from the channel. Thus, it provides an upper bound on the smallest $v$ such that $\Lambda$ is coherence-breaking. The procedure is summarised in Algorithm \ref{algorithm_search}.

\begin{algorithm}
	\caption{Coherence certification}
	\label{algorithm_search}
	\KwIn{Random $\mathcal{E} = \lbrace \psi_{x} \rbrace_{x=1}^{N}$, $v_0 = 0$, $\epsilon=10^{-6}$}
	
	\textit{Step} 1 (\textbf{SDP hierarchy}): $(\bar{v}(\mathcal{E}), \lbrace \beta_x \rbrace_{x=1}^{N}, W) \leftarrow \mathcal{E}$
	
	\textit{Step} 2 (\textbf{Witness violation}): $\displaystyle \mathcal{E} =  \arg \min_{\mathcal{E}} \, W(\Lambda(\mathcal{E}))$
	
	\textit{Iteration Step} \While{$|\bar{v}-v_0| > \epsilon$}{
		
		$v_0 = \bar{v}$;
		
		$(v^{*}(\mathcal{E}), \lbrace \beta_x \rbrace_{x=1}^{N}, W) \leftarrow \mathcal{E}$;
		
		$\displaystyle \mathcal{E} =  \arg \min_{\mathcal{E}} \, W(\Lambda(\mathcal{E}))$;
	}
	
	\Return{$\bar{v},\mathcal{E}$}
	
\end{algorithm}

\noindent Our implementation is available at \cite{github-code}. As a test example, we considered the depolarisation channel in dimension $d=20$ which is known to be coherence breaking if and only if $v \leq 0.1367$ \cite{Cobucci2026}. By fixing $N=5$, the procedure returns an upper bound $\bar{v} = 0.3778$, while for $N=10$, we get $\bar{v} = 0.2983$. As expected, the bound decreases for larger values of $N$ and approaches the known limit. Additionally, we considered the multilevel amplitude-damping (MAD) channel \cite{Chessa2021}, whose Kraus operators are:
\begin{equation}
	\begin{aligned}
		K_0 &= \ketbra{0}{0} + \sum_{j=1}^{d-1}\sqrt{1-\xi_j}\ketbra{j}{j},\\
		K_{ij} &= \sqrt{\gamma_{ji}}\ketbra{j}{i}, \quad \text{for } 0 \leq i < j \leq d-1,
	\end{aligned}
\end{equation}
for a set of coefficients $\gamma_{ji} \leq 1$ and $\xi_j = \sum_{0\leq i < j}\gamma_{ji}$. Fixing the dimension to $d=5$, the number of states to $N=30$ and the values of $\gamma_{ji} = 0.5$ for all $i,j$, the procedure returns $\bar{v} = 0.4641$.

The above algorithm is motivated for systems of dimension higher than qubit. For qubit channels, we employ the tailor-made method from section~\ref{Sec:qubit}. Importantly, this method does not only certify that $\Lambda$ preserves coherence, but it can also determine when $\Lambda$ is CB. For the latter, we must show that the image of $\Lambda$, namely $\mathcal{E}_\Lambda$, admits an incoherent model. This can be achieved by selecting $\mathcal{E}$ as a set of Hermitian matrices so that the convex hull of their Bloch vectors contains the Bloch sphere. After these matrices are sent through $\Lambda$, they yield an outer polyhedral approximation of $\mathcal{E}_\Lambda$. If the matrices obtained after the channel admit an incoherent model, it follows that $\Lambda$ is CB. As a simple example, we illustrate this for the case of the depolarisation channel, $\Lambda^{\text{depol}}_v$. Recall from the discussion in section~\ref{Sec:CBC} that it is CB if and only if $v\leq \frac{1}{2}$. We can verify this up to high accuracy with the linear programming relaxations. To this end, we have considered a polyhedron whose $N=1012$ vertices define the states of the set for which we want to test coherence. The relaxations are based on a second polyhedron with 220 vertices, used as inner approximation of the Bloch sphere, from which we obtain upper and lower bounds on the critical value of $v$. We find that the channel is CB when $v\lesssim  0.4999$ and that it is coherence-preserving when $v\gtrsim 0.5029$. This is already nearly the exact value, and the gap can be shrunk even further by considering a more expensive polyhedral approximation of the Bloch sphere.

\begin{figure}[t!]
	\centering
	\includegraphics[width=0.5\linewidth]{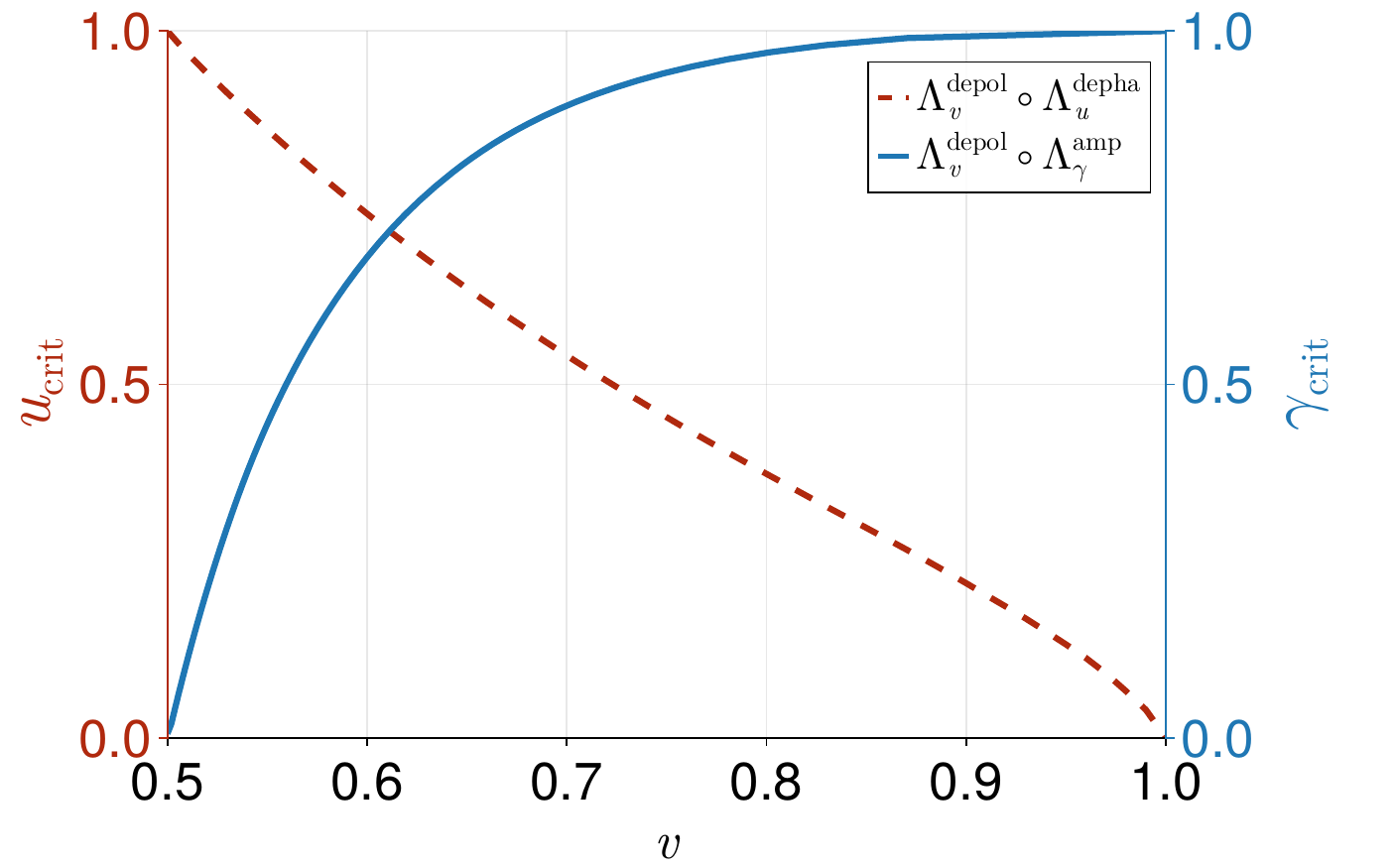}
	\centering
	\caption{Bounds on coherence properties in the channels $\Lambda_v^{\text{depol}}\circ \Lambda^{\text{depha}}_{u}$ (red) and $\Lambda_v^{\text{depol}}\circ \Lambda^{\text{amp}}_{\gamma}$ (blue). We have computed  bounds both for the channels being coherence-preserving and coherence-breaking, but only one is displayed since they are nearly identical (due to the polyhedron's shrinking factor being $r\approx 0.9879$). For any given depolarising parameter $v$, channels with dephasing parameter $u$ below the red curve are coherence breaking, whereas those with values of the dampening parameter $\gamma$ below the blue line are coherence preserving.}
	\label{Fig:qubits}
\end{figure}

We can now use this method to address the coherence properties of channels for which no analytical solution is known. We focus on two case studies. Firstly, we compose the depolarisation channel, $\Lambda_v^{\text{depol}}$, and the dephasing channel, $\Lambda^{\text{depha}}_{u}(\rho_x)=u \rho_x +(1-u)\,\mathrm{Diag}(\rho_x)$, where $\mathrm{Diag}(\rho_x)$ is $\rho_x$ when all off-diagonals are set to zero. Thus, we have $\Lambda=\Lambda_v^{\text{depol}}\circ \Lambda^{\text{depha}}_{u}$. Secondly, we compose the depolarisation channel with the amplitude-damping channel, $\Lambda^{\text{amp}}_{\gamma}(\rho)=K_0 \rho K_0^\dagger + K_1 \rho K_1^\dagger$, with Kraus operators $K_0=\ketbra{0}+\sqrt{1-\gamma}\ketbra{1}$ and $K_1=\sqrt{\gamma}\ketbra{0}{1}$. This gives $\Gamma=\Lambda_v^{\text{depol}}\circ \Lambda^{\text{amp}}_{\gamma}$. For each value of $v$, we have computed upper and lower bounds on the critical value of $u$ and $\gamma$ at which the channels are CB, testing a set of $N=220$ states. This is again based on the method of polyhedral approximations  discussed in section~\ref{Sec:qubit}. In this case, the tested states and the polyhedral approximation correspond to the same 220 vertices. In the dephasing case, the value of $v$ is fixed at each point and the optimisation is performed over $u$. Similarly, in the amplitude-damping case, we fix $\gamma$ and compute the optimisation over $v$. The critical values $u_{\text{crit}},v_{\text{crit}}$ and $\gamma_{\text{crit}}$ at which the channel is certified as CB are plotted in Fig~\ref{Fig:qubits}. These values corresponds to the inner polyhedral approximation of the Bloch sphere. As in section~\ref{Sec:qubit}, we can elongate the Bloch vectors to obtain the outer polyhedral approximation in order to also obtain bounds for which the channel preserves coherence. These bounds are without exception very close to each other and we therefore do not display them in Fig~\ref{Fig:qubits}.

Although the same method applies beyond qubits, this comes with a substantial loss in precision. Here, we investigated the critical visibility for which a coherent model exists under isotropic noise in dimension $d=3$. For this we use a 291-vertex polytope obtained by starting from the 12 vertices corresponding to MUBs and iteratively extending the polytope by adding projectors associated with the facets tangent to the inscribed hypersphere, as outlined in \cite{Porto2025}. This corresponds to a shrinking factor $r\approx 0.8177$ and certify that for any $v<0.4160$ the channel is CB, whereas for $v>0.5088$ it preserves coherence.
Generalizing this method to higher dimensions requires polytopes with a number of vertices that rapidly grows beyond tractable size.

\end{document}